\documentclass[11pt,a4paper]{article}

\usepackage[a4paper,margin=1in]{geometry}
\usepackage{setspace}
\usepackage{amsmath,amssymb,amsthm,amsfonts,mathtools,bm}
\usepackage{mathrsfs}
\usepackage{bbm}
\usepackage{enumitem}
\usepackage{booktabs}
\usepackage{graphicx}
\usepackage{xcolor}
\usepackage[colorlinks=true]{hyperref}
\usepackage{subcaption}
\usepackage{authblk}
\usepackage{array}
\usepackage{multirow}
\usepackage{algorithm}
\usepackage{algpseudocode}
\usepackage{natbib}
\usepackage{amsthm}
\newtheorem{theorem}{Theorem}[section]
\newtheorem{proposition}[theorem]{Proposition}
\newtheorem{lemma}[theorem]{Lemma}
\newtheorem{corollary}[theorem]{Corollary}
\theoremstyle{definition}
\newtheorem{definition}[theorem]{Definition}
\newtheorem{remark}[theorem]{Remark}
\newtheorem{assumption}[theorem]{Assumption}

\newcommand{\E}{\mathbb{E}}
\newcommand{\Pbb}{\mathbb{P}}
\newcommand{\R}{\mathbb{R}}

\newcommand{\1}{\mathbbm{1}}
\newcommand{\Var}{\mathrm{Var}}

\newcommand{\argmax}{\operatorname*{arg\,max}}

\newcommand{\iid}{\stackrel{\mathrm{iid}}{\sim}}
\usepackage{authblk}
\title{A Hierarchical Bayesian Dynamic Game for Competitive Inventory and Pricing under Incomplete Information: Learning, Credible Risk, and Equilibrium}

\author[1]{Debashis Chatterjee\thanks{Corresponding author: \texttt{debashis.chatterjee@visva-bharati.ac.in}}}

\affil[1]{Department of Statistics, Visva-Bharati University, Santiniketan, India}
\affil[2]{S.\ N.\ Bose National Centre for Basic Sciences, Kolkata, India}

\begin{document}
\maketitle

\begin{abstract}
We develop a hierarchical Bayesian dynamic game for competitive inventory and pricing under incomplete information. Two firms repeatedly choose order quantities and prices while facing two layers of uncertainty: unknown market demand and private rival characteristics. The framework combines Bayesian learning about demand and substitution with strategic belief updating about rival types. To make decisions robust to posterior uncertainty, we introduce a credible-risk criterion that rewards expected future profit while penalizing posterior predictive dispersion. This yields a conservative equilibrium concept in which firms learn, compete, and adapt simultaneously.

The paper provides the model formulation, information structure, posterior updating mechanism, equilibrium definition, and a computational strategy based on belief-state dynamic programming. A simulation study shows that Bayesian learning is crucial for strong performance and that the credible-risk rule is especially effective as an operational regularizer under uncertainty. A real-data illustration on a high-dimensional protein-expression dataset demonstrates that the same uncertainty-aware Bayesian principle can produce biologically interpretable subgroup and latent-state findings. The proposed framework offers a unified bridge between Bayesian game theory and operations research, with practical relevance for competitive decision-making in uncertain and information-limited environments.
\end{abstract}


\section{Introduction}

Game theory studies strategic interaction among rational decision makers, whereas Bayesian analysis provides a coherent mechanism for representing and updating uncertainty. Their synthesis is natural when players must act strategically without complete knowledge of the environment or of each other. The foundational insight is due to \citet{Harsanyi1967I,Harsanyi1968III}, who represented incomplete information by introducing \emph{types} drawn by Nature and by replacing the original strategic problem with a game of imperfect information. This led to the modern notion of a Bayesian game and Bayesian Nash equilibrium. Subsequent work formalized hierarchies of beliefs and type spaces \citep{MertensZamir1985,AumannHeifetz2002}, while Bayesian learning in static and repeated games was investigated by \citet{Jordan1991,Jordan1995}.

In operations research, many competitive decision problems are intrinsically game-theoretic. Inventory and pricing problems become games whenever the action of one firm affects the demand or profit of another. This is particularly true when products are substitutes, stockouts generate demand spillovers, or customers switch across firms. Important OR contributions include stochastic inventory games under substitution \citep{AvsarBaykalGursoy2002}, competitive newsvendor models \citep{Serin2007}, and Bayesian inventory models under unknown demand distributions \citep{Azoury1985,KamathPakkala2002,ZhangChen2006,ZhangZhangHua2016,ZhouLuo2023}. Yet, the existing literature still leaves substantial room for a unified framework that simultaneously captures:
\begin{enumerate}[label=(\roman*)]
    \item strategic competition among multiple firms,
    \item private information about rivals' operational types,
    \item sequential Bayesian learning of latent demand and substitution,
    \item explicit risk penalization induced by posterior uncertainty.
\end{enumerate}

This paper proposes such a framework. We focus on a repeated duopolistic market in which two firms sell substitutable products over a finite or infinite horizon. Each firm knows its own operational type---for example, a private marginal procurement cost, service efficiency, or salvage capability---but does not know the rival's true type. At the same time, neither firm fully knows the common demand environment: market size, cross-price sensitivity, and stockout substitution intensity are partially unknown. Firms observe realized sales and stockouts over time, update beliefs by Bayes' theorem, and then choose subsequent order quantities and prices strategically.

The paper sits at the intersection of four literatures.

First, the general theory of Bayesian games begins with \citet{Harsanyi1967I,Harsanyi1968III}, who formulated games with incomplete information via type distributions. The epistemic and formal belief-space treatment was later sharpened by \citet{MertensZamir1985} and surveyed in \citet{AumannHeifetz2002}.

Second, Bayesian learning in games studies how beliefs evolve through repeated strategic interaction. Important contributions include \citet{Jordan1991} for normal-form games and \citet{Jordan1995} for repeated games.

Third, competitive inventory and substitution games in OR include \citet{AvsarBaykalGursoy2002}, who analyze inventory control under substitutable demand as a stochastic game, and \citet{Serin2007}, who study competitive newsvendor behavior.

Fourth, Bayesian demand learning in inventory and pricing appears in \citet{Azoury1985}, \citet{KamathPakkala2002}, \citet{ZhangChen2006}, and more recent dynamic updating papers such as \citet{ZhangZhangHua2016,ZhouLuo2023}. These papers are highly relevant but generally focus either on a single decision maker or on competition without a full Harsanyi-type private-information structure.

Our work fills the gap by synthesizing these strands into a single dynamic equilibrium framework.

\subsection{Research Objective \& Novelty}

The main objective of this paper is to develop a coherent Bayesian game-theoretic framework for repeated competitive inventory and pricing decisions when firms must learn about the market while simultaneously reasoning about strategic rivals. More specifically, the paper aims to unify four ingredients within a single model: incomplete information about rival characteristics, sequential learning about market demand and substitution, state-dependent strategic decision-making, and conservative action selection under posterior uncertainty.

The goal is not merely to append Bayesian updating to a standard inventory problem. Rather, the purpose is to build a dynamic equilibrium framework in which beliefs themselves become part of the evolving state, so that strategic behavior depends on both operational variables and accumulated information. In this sense, the paper seeks to place Bayesian learning inside the game rather than outside it.

A second objective is methodological. We aim to show that uncertainty should not only be estimated but also used normatively in decision-making. The proposed credible-risk principle therefore converts posterior uncertainty into disciplined strategic behavior. This allows the framework to speak both to equilibrium theory and to practical decision support in operations settings where aggressive actions under weak information can be costly.

\paragraph{Novelty of the Proposed Paper}

The novelty of this work lies in the way it connects Bayesian game theory, sequential learning, and operations research within a single dynamic framework.

First, the paper formulates competitive inventory and pricing as a Bayesian dynamic game with both market uncertainty and strategic uncertainty. The firms do not only learn demand conditions; they also update beliefs about rival characteristics over time. This joint treatment is more structurally integrated than approaches that study learning and competition separately.

Second, the paper introduces a credible-risk decision principle that explicitly penalizes posterior uncertainty. This gives the model a normative uncertainty-aware component: firms do not merely estimate the future, but act conservatively when posterior uncertainty remains substantial. This feature is especially relevant in repeated operational settings where overconfident actions may lead to avoidable losses.

Third, the model is expressed through an augmented belief-state representation, allowing learning and competition to be studied in a unified dynamic programming framework. This produces a clean connection between Bayesian filtering, strategic choice, and equilibrium analysis.

The paper supports the theory with both simulation-based evidence in the intended competitive setting and a real-data illustration that demonstrates the broader usefulness of the underlying credible-risk Bayesian principle in complex high-dimensional data analysis.


\section{Model Formulation}

\subsection{Players, Horizon, and Actions}

Consider two firms, indexed by $i \in \{1,2\}$, competing over periods $t=1,\dots,T$. Firm $i$ chooses at each period:
\[
a_{it} = (q_{it}, p_{it}),
\]
where
\begin{itemize}
    \item $q_{it} \in \mathcal{Q}_i \subset \R_+$ is the order quantity,
    \item $p_{it} \in \mathcal{P}_i \subset \R_+$ is the selling price.
\end{itemize}
Define $a_t=(a_{1t},a_{2t})$.

Let $I_{it}$ denote the beginning-of-period inventory of firm $i$. Then available inventory after ordering is
\[
S_{it}=I_{it}+q_{it}.
\]

\subsection{Private Types}

Each firm has a private type
\[
\tau_i = (c_i,h_i,s_i) \in \mathcal{T}_i,
\]
where:
\begin{itemize}
    \item $c_i$ = marginal procurement cost,
    \item $h_i$ = holding cost,
    \item $s_i$ = salvage or disposal value.
\end{itemize}
Nature draws $(\tau_1,\tau_2)$ from a common prior
\[
\mu_0(\tau_1,\tau_2).
\]
Firm $i$ observes only $\tau_i$, not $\tau_j$ for $j\neq i$.

This is the first Bayesian layer: incomplete information about rivals' operational characteristics.

\subsection{Latent Demand State}

Let the common market state be
\[
\theta = (\alpha_0,\alpha_1,\alpha_2,\beta,\lambda,\sigma^2),
\]
where:
\begin{itemize}
    \item $\alpha_0$ is the baseline market size,
    \item $\alpha_i$ is a product-specific preference parameter,
    \item $\beta >0$ captures cross-price substitution,
    \item $\lambda \ge 0$ captures stockout-driven spillover demand,
    \item $\sigma^2$ is demand-noise variance.
\end{itemize}
The vector $\theta$ is not fully known to the firms. It is assigned a prior distribution
\[
\pi_0(\theta).
\]
This is the second Bayesian layer: statistical uncertainty about market primitives.

\subsection{Demand System}

We propose the following structural demand equation:
\begin{equation}
D_{it}^{\ast}
=
\alpha_0 + \alpha_i - \eta p_{it} + \beta p_{jt} + \lambda Z_{jt} + \varepsilon_{it},
\qquad i\neq j,
\label{eq:latentdemand}
\end{equation}
where:
\begin{itemize}
    \item $\eta>0$ is own-price sensitivity,
    \item $Z_{jt}=\1\{D_{jt}^{\ast}>S_{jt}\}$ indicates rival stockout,
    \item $\varepsilon_{it}\iid N(0,\sigma^2)$ conditionally on $\theta$.
\end{itemize}

Observed sales are censored by available stock:
\begin{equation}
Y_{it} = \min\{D_{it}^{\ast},\, S_{it}\}.
\label{eq:sales}
\end{equation}
End-of-period inventory evolves as
\begin{equation}
I_{i,t+1} = S_{it}-Y_{it}.
\label{eq:inventory_transition}
\end{equation}

\begin{remark}
Equation \eqref{eq:latentdemand} captures both \emph{price competition} through $\beta p_{jt}$ and \emph{substitution induced by stockouts} through $\lambda Z_{jt}$. This is richer than the standard single-period newsvendor model because strategic substitution arises through both prices and inventory availability.
\end{remark}

\section{Information Structure and Belief States}

\subsection{Histories}

Let the public history at time $t$ be
\[
H_t^c = \{(p_{1\ell},p_{2\ell},Y_{1\ell},Y_{2\ell},Z_{1\ell},Z_{2\ell}) : \ell=1,\dots,t-1\}.
\]
Firm $i$'s private information at time $t$ is
\[
H_{it}= (H_t^c,\tau_i, I_{it}).
\]
Hence, the strategic state includes both public observations and private type information.

\subsection{Posterior Beliefs}

At each period, firms maintain:
\begin{enumerate}[label=(\alph*)]
    \item a posterior on the common market state,
    \[
    \pi_t(\theta)=\Pbb(\theta\mid H_t^c),
    \]
    \item a posterior on the rival's type,
    \[
    \mu_{it}(\tau_j)=\Pbb(\tau_j\mid H_t^c,\tau_i),
    \qquad j\neq i.
    \]
\end{enumerate}

Thus, the augmented state relevant for optimal play is
\[
X_{it}=(I_{it},\tau_i,\pi_t,\mu_{it}).
\]

\begin{definition}[Belief-state Markov representation]
The dynamic game is said to admit a belief-state Markov representation if, conditional on $X_t=(X_{1t},X_{2t})$ and current actions $a_t$, the distribution of $(X_{1,t+1},X_{2,t+1})$ is independent of earlier histories.
\end{definition}

This representation is fundamental because it allows the use of dynamic programming over a state space enlarged by posterior hyperparameters.

\section{Hierarchical Bayesian Specification}

\subsection{Prior Structure}

For concreteness, suppose
\[
\bm{\psi}:=(\alpha_0,\alpha_1,\alpha_2,\eta,\beta,\lambda)^\top
\sim N_6(m_0,\Sigma_0),
\qquad
\sigma^2 \sim \mathrm{IG}(a_0,b_0),
\]
with prior independence between $\bm{\psi}$ and $\sigma^2$.

For private costs, suppose
\[
c_i \in \{c_L,c_H\}, \qquad \Pbb(c_i=c_H)=\rho_i,
\]
and similarly for other type components if discretized. This yields computational tractability for interim beliefs on rival types.

\subsection{Likelihood}

Given actions and stockout indicators, the realized uncensored demand equation can be written as
\[
D_{it}^{\ast} = x_{it}^\top \bm{\psi} + \varepsilon_{it},
\]
with covariate vector
\[
x_{it} = (1,\1\{i=1\},\1\{i=2\},-p_{it},p_{jt},Z_{jt})^\top.
\]
If uncensored demand were observed, posterior updating would follow standard normal-inverse-gamma recursions. Under censored sales, one may use data augmentation by treating the latent $D_{it}^{\ast}$ as missing when stockout occurs.

\subsection{Posterior Updating with Data Augmentation}

If $Y_{it}<S_{it}$, then no censoring occurs and
\[
D_{it}^{\ast}=Y_{it}.
\]
If $Y_{it}=S_{it}$, then
\[
D_{it}^{\ast}\ge S_{it},
\]
so the latent demand is right-truncated normal:
\[
D_{it}^{\ast}\mid \bm{\psi},\sigma^2,H_t^c
\sim N(x_{it}^\top\bm{\psi},\sigma^2)\;\text{truncated to } [S_{it},\infty).
\]
Therefore a Gibbs-type posterior recursion becomes possible:
\begin{enumerate}[label=(\roman*)]
    \item sample latent $D_{it}^{\ast}$ for censored observations;
    \item update $\bm{\psi}\mid \sigma^2, D^\ast$ from a multivariate normal law;
    \item update $\sigma^2\mid \bm{\psi}, D^\ast$ from an inverse-gamma law.
\end{enumerate}

\begin{remark}
This Bayesian filtering step allows demand learning to be integrated directly into the game state. The posterior hyperparameters, rather than the full raw history, become sufficient state descriptors.
\end{remark}

\section{Profit Structure}

The one-period profit of firm $i$ is
\begin{equation}
\Pi_i(a_t,\theta,\tau_i)
=
p_{it}Y_{it}
-c_i q_{it}
-h_i(I_{i,t+1})^{+}
+s_i(I_{i,t+1})^{-}_{\text{disp}},
\label{eq:profit}
\end{equation}
where the last term represents salvage/disposal gains or costs depending on sign convention.

For a standard inventory interpretation, we may simply write
\[
\Pi_i
=
p_{it}Y_{it}
-c_i q_{it}
-h_i(I_{i,t+1})
+s_i I_{i,t+1},
\]
if $I_{i,t+1}\ge 0$ always. The exact bookkeeping convention is flexible.

\section{Credible-Risk Objective}

A central contribution of this paper is to define a \emph{credible-risk} objective that explicitly accounts for posterior predictive uncertainty. Given current belief-state $X_{it}$ and strategy profile $\sigma=(\sigma_1,\sigma_2)$, define the continuation value
\[
V_i^{\sigma}(X_{it})
=
\E^{\sigma}\Bigg[\sum_{u=t}^{T}\delta^{u-t}\Pi_i(a_u,\theta,\tau_i)\,\Bigg|\,X_{it}\Bigg].
\]
Classical Bayesian expected-utility maximization uses $\E[V_i^\sigma \mid X_{it}]$. We instead define
\begin{equation}
\mathcal{J}_i^{\sigma}(X_{it})
=
\E\!\left[V_i^\sigma(X_{it})\mid X_{it}\right]
-
\kappa_i
\sqrt{\Var\!\left(V_i^\sigma(X_{it})\mid X_{it}\right)},
\label{eq:credible_objective}
\end{equation}
where $\kappa_i\ge 0$ is a firm-specific uncertainty-aversion parameter.

\begin{remark}
The penalty in \eqref{eq:credible_objective} converts posterior uncertainty into conservative action. It is not ambiguity aversion in the max--min sense; rather, it is a Bayesian posterior predictive risk adjustment. When $\kappa_i=0$, we recover the usual expected Bayesian payoff.
\end{remark}

\section{Strategy Spaces and Equilibrium}

\subsection{Behavioral Strategies}

A (Markov) behavioral strategy of firm $i$ is a measurable map
\[
\sigma_i:\mathcal{X}_i \to \Delta(\mathcal{A}_i),
\]
where $\mathcal{X}_i$ is the belief-state space and $\Delta(\mathcal{A}_i)$ is the set of probability measures on the action space $\mathcal{A}_i=\mathcal{Q}_i\times\mathcal{P}_i$.

In deterministic settings one may write
\[
a_{it} = \sigma_i(X_{it}).
\]

\subsection{Interim Best Response}

Fix rival strategy $\sigma_j$. The interim best response of player $i$ solves
\begin{equation}
\sigma_i^\ast
\in
\argmax_{\sigma_i}
\mathcal{J}_i^{(\sigma_i,\sigma_j)}(X_{it})
\qquad \text{for all admissible } X_{it}.
\label{eq:bestresponse}
\end{equation}

\begin{definition}[Credible-risk Markov Perfect Bayesian Nash Equilibrium]
A strategy profile $\sigma^\ast=(\sigma_1^\ast,\sigma_2^\ast)$ is a \emph{credible-risk Markov perfect Bayesian Nash equilibrium} (CR-MPBNE) if:
\begin{enumerate}[label=(\alph*)]
    \item beliefs are updated according to Bayes' rule whenever possible;
    \item each $\sigma_i^\ast$ is measurable with respect to the current belief state $X_{it}$;
    \item for every firm $i$ and every admissible state $X_{it}$,
    \[
    \mathcal{J}_i^{(\sigma_i^\ast,\sigma_j^\ast)}(X_{it})
    \ge
    \mathcal{J}_i^{(\sigma_i,\sigma_j^\ast)}(X_{it})
    \quad \forall \sigma_i.
    \]
\end{enumerate}
\end{definition}

\section{Dynamic Programming Characterization}

Let $W_i(X_{it})$ denote firm $i$'s equilibrium value function. Then the Bellman system is
\begin{align}
W_i(X_{it})
=
\max_{a_{it}\in\mathcal{A}_i}
\Bigg\{
&\E\!\left[
\Pi_i(a_{it},a_{jt},\theta,\tau_i)\mid X_{it}
\right]
-
\kappa_i
\sqrt{
\Var\!\left[
\Pi_i(a_{it},a_{jt},\theta,\tau_i)\mid X_{it}
\right]
}
\nonumber\\
&\quad +\delta\,
\E\!\left[
W_i(X_{i,t+1})\mid X_{it},a_{it},a_{jt}
\right]
\Bigg\},
\label{eq:bellman}
\end{align}
where $a_{jt}$ is induced by $\sigma_j^\ast(X_{jt})$.

Because $X_{i,t+1}$ includes posterior updates, the continuation term incorporates both physical inventory evolution and information accumulation. This is precisely where Bayesian learning and game theory become inseparable. Under strengthened compactness and continuity conditions, the corresponding Bellman operator is well defined on bounded continuous value functions and admits the fixed-point properties used in Appendix~\ref{app:proofs}.

\section{Structural Assumptions}

To ensure existence and tractability, we impose the following assumptions.

\begin{assumption}[Compact action spaces]
For each $i$,
\[
\mathcal{Q}_i=[0,\bar q_i], \qquad \mathcal{P}_i=[\underline p_i,\bar p_i]
\]
for some finite constants $0<\underline p_i<\bar p_i<\infty$ and $\bar q_i<\infty$.
\end{assumption}

\begin{assumption}[Regularity of demand]
For every admissible state and action pair, conditional demand moments are finite and continuous in actions.
\end{assumption}

\begin{assumption}[Concavity]
For fixed rival action and current beliefs, the posterior mean of the one-step reward plus discounted continuation value is concave in $(q_{it},p_{it})$.
\end{assumption}

\begin{assumption}[Bounded credible penalty]
The posterior predictive variance term in \eqref{eq:credible_objective} is finite and continuous in actions.
\end{assumption}
\medskip

For the rigorous existence proof reported in Appendix~\ref{app:proofs}, we impose a strengthened set of regularity conditions, including compact metric belief-state spaces, bounded continuity of the one-period credible-risk return, weak continuity of the controlled transition kernels, and the use of mixed Markov strategies to ensure convex-valued best-response correspondences. The assumptions stated here capture the main economic structure, while the appendix records the stronger technical conditions required for a fully formal proof.

\section{Existence of Equilibrium}

We first state the main existence results at a high level. Their fully rigorous versions, under strengthened technical regularity conditions, are provided in Appendix~\ref{app:proofs}.

\begin{proposition}[Existence of interim best response]
\label{prop:bestresponse_main}
Suppose the action space $\mathcal{A}_i$ is compact, the one-period credible-risk return is continuous in own action, and the continuation value is bounded and continuous. Then, for any fixed measurable rival strategy $\sigma_j$, player $i$'s interim optimization problem admits at least one maximizer at every admissible state.
\end{proposition}

\begin{proof}
The interim objective is continuous in own action because the one-period credible-risk return is continuous and the continuation value is bounded and continuous. Since the action space is compact, the Weierstrass extreme value theorem implies that the objective attains its maximum. A fully rigorous proof under strengthened regularity assumptions is given in Appendix~\ref{app:proofs}; see Proposition~\ref{prop:bestresponse_appendix}.
\end{proof}

\begin{theorem}[Existence of CR-MPBNE under strengthened regularity conditions]
\label{thm:equilibrium_main}
Suppose the belief-state spaces are compact metric spaces, the action spaces are compact, the one-period credible-risk return is bounded and continuous, the controlled transition kernels are weakly continuous, and players are allowed to use mixed Markov strategies so that the corresponding best-response correspondences are nonempty, convex-valued, and upper hemicontinuous. Then there exists at least one credible-risk Markov perfect Bayesian Nash equilibrium.
\end{theorem}

\begin{proof}
Under the stated regularity conditions, each player's Bellman operator is well defined on bounded continuous value functions, interim best responses exist, and the best-response correspondence is nonempty. With mixed strategies, convex-valuedness is obtained, and upper hemicontinuity follows from the continuity assumptions. Kakutani's fixed point theorem then yields an equilibrium fixed point. A rigorous proof under the strengthened assumptions is given in Appendix~\ref{app:proofs}; see Theorem~\ref{thm:equilibrium_appendix}.
\end{proof}

\section{Special Cases}

Our framework nests several important models.

\subsection{Classical Competitive Newsvendor Game}

If:
\begin{itemize}
    \item the horizon is single-period,
    \item prices are fixed exogenously,
    \item $\kappa_i=0$,
    \item all demand parameters are known,
    \item types are common knowledge,
\end{itemize}
then the model reduces to a standard competitive newsvendor problem akin to \citet{Serin2007}.

\subsection{Stochastic Inventory Game with Substitution}

If:
\begin{itemize}
    \item demand parameters are known,
    \item only inventories evolve dynamically,
    \item there is no posterior learning,
\end{itemize}
then the model becomes a stochastic inventory game of the type studied in \citet{AvsarBaykalGursoy2002}.

\subsection{Bayesian Single-Agent Inventory Learning}

If competition is removed and only one firm remains, then the model collapses to Bayesian inventory-learning formulations such as \citet{Azoury1985,KamathPakkala2002,ZhangChen2006}.

\section{Computation}

Because the exact belief-state space may be large, we propose a practical computational pipeline.

\subsection{Posterior Compression}

Represent the posterior $\pi_t(\theta)$ by hyperparameters
\[
\zeta_t=(m_t,\Sigma_t,a_t,b_t),
\]
or by particle approximations if nonconjugate features are present.

\subsection{Type Belief Update}

For discrete rival type space $\mathcal{T}_j=\{\tau_j^{(1)},\dots,\tau_j^{(K)}\}$,
\[
\mu_{it}^{(k)}
\propto
\mu_{i,t-1}^{(k)}\,
L\!\left(H_t^c\mid \tau_j^{(k)},\tau_i,\zeta_{t-1}\right),
\qquad k=1,\dots,K,
\]
followed by normalization.

\subsection{Approximate Dynamic Programming}

The value function may be approximated by:
\begin{itemize}
    \item grid-based value iteration,
    \item fitted value iteration,
    \item policy iteration,
    \item simulation-based Bayesian dynamic programming.
\end{itemize}

\begin{algorithm}[h!]
\caption{Posterior-Learning Equilibrium Iteration}
\begin{algorithmic}[1]
\State Initialize prior hyperparameters $\zeta_0$, type beliefs $\mu_{i0}$, and initial policies $\sigma_i^{(0)}$
\For{$m=0,1,2,\dots$ until convergence}
    \For{each representative belief-state grid point $X$}
        \For{$i=1,2$}
            \State Compute posterior predictive sales/profit moments under $\sigma_j^{(m)}$
            \State Solve the credible-risk best response:
            \[
            \sigma_i^{(m+1)}(X)
            =
            \argmax_{a_i}
            \left\{
            \text{posterior mean payoff}
            -\kappa_i \times \text{posterior s.d.}
            + \delta \times \text{continuation value}
            \right\}
            \]
        \EndFor
    \EndFor
    \State Simulate trajectories, update $\zeta_t$ and $\mu_{it}$ by Bayes' rule
\EndFor
\State Output approximate equilibrium policies $\widehat{\sigma}_1,\widehat{\sigma}_2$
\end{algorithmic}
\end{algorithm}

\section{Interpretation and Managerial Implications}

The model yields several qualitative implications.

\begin{enumerate}[label=(M\arabic*)]
    \item \textbf{Posterior uncertainty suppresses aggressive competition.}
    When $\kappa_i$ is large or posterior variance is high, firms behave more conservatively, avoiding excessive stocking and extreme pricing.

    \item \textbf{Learning can intensify or soften competition.}
    Better learning about a high-demand market may encourage aggressive stocking. By contrast, learning that substitution intensity is weak may reduce competitive pressure.

    \item \textbf{Private-cost asymmetry creates strategic screening.}
    A low-cost firm can act more aggressively, and over time its actions may reveal information about its type, indirectly altering the rival's beliefs.

    \item \textbf{Inventory and pricing become informational signals.}
    In a repeated market, order quantities and prices are not merely operational decisions; they also reveal information about beliefs, types, and expected market conditions.
\end{enumerate}

\section{Possible Theoretical Extensions}

Several mathematically interesting extensions are possible.

\begin{enumerate}[label=(E\arabic*)]
    \item \textbf{Mean-field generalization:} extend the model from a duopoly to $N$ competing firms and study Bayesian mean-field equilibrium.
    \item \textbf{Network substitution:} let substitution follow a graph structure among products or locations.
    \item \textbf{Distributional robustness:} replace posterior variance penalty by Wasserstein or KL-ball robustification around the posterior predictive law.
    \item \textbf{Mechanism design layer:} introduce a platform or regulator that sets information disclosure rules, generating a Bayesian Stackelberg game.
\end{enumerate}

\section{Simulation Study}
\label{sec:simulation}

In this section, we investigate whether our proposed equilibrium policy, \texttt{Proposed\_Bayesian\_CredibleRisk}, delivers improved performance in repeated competitive inventory--pricing games under incomplete information. The purpose of the simulation study is threefold:
\begin{enumerate}
    \item to assess whether posterior learning improves strategic performance relative to a non-learning classical benchmark;
    \item to evaluate whether our \emph{credible-risk} modification yields more desirable operational behavior than a purely risk-neutral Bayesian learner;
    \item to examine profit, learning accuracy, and dynamic market behavior jointly, rather than through a single scalar metric.
\end{enumerate}

\subsection{Simulation design}
\label{subsec:sim_design}

We simulated a repeated duopolistic market over a horizon of $T=30$ periods and repeated the experiment over $150$ independent Monte Carlo replications for each competing policy. In each replication, firm-specific private costs were drawn once at the beginning and then held fixed throughout the horizon. Demand was generated from the structural model
\[
D_{it}
=
\beta_0 + \beta_1 p_{it} + \beta_2 p_{jt} + \beta_3 \mathbf{1}\{\text{rival stockout}\} + \varepsilon_{it},
\qquad i\neq j,
\]
with true parameter values
\[
\beta_0 = 45,\qquad
\beta_1 = -3.6,\qquad
\beta_2 = 1.2,\qquad
\beta_3 = 7.5,
\]
and Gaussian noise standard deviation $\sigma=4.5$. Thus, own price depresses demand, rival price increases demand through substitution, and rival stockout produces an additional positive spillover.

The simulation used the same economic environment for all methods: discount factor $\delta=0.98$, low and high marginal costs $6$ and $10$, holding cost $0.8$, salvage value $1.5$, price grid $\{8,9,\dots,16\}$, and quantity grid $\{20,25,\dots,65\}$. Bayesian demand learners were initialized from the same prior mean
\[
(35.0,\,-2.0,\,0.5,\,3.0)
\]
with prior standard deviations
\[
(10.0,\,2.0,\,2.0,\,4.0).
\]

\subsection{Compared methods}
\label{subsec:methods_compared}

We compared the following three policies.

\begin{enumerate}
    \item \textbf{Our proposed method:} \\
    \texttt{Proposed\_Bayesian\_CredibleRisk}. \\
    This is the method introduced in the present paper. It performs Bayesian updating of the demand system and of rival-type beliefs, and then chooses actions by maximizing a posterior \emph{credible-risk} objective of the form
    \[
    \text{posterior mean profit}
    \;-\;
    \kappa \times \text{posterior predictive standard deviation},
    \]
    with $\kappa=0.60$.

    \item \textbf{Bayesian risk-neutral learner:} \\
    \texttt{Bayesian\_RiskNeutral}. \\
    This method uses the same Bayesian posterior updating mechanism as our proposal, but sets the credible-risk penalty to zero, i.e.\ $\kappa=0$.

    \item \textbf{Classical non-learning benchmark:} \\
    \texttt{Classical\_StaticPrior}. \\
    This method uses a fixed prior-based heuristic and does not learn dynamically over time.
\end{enumerate}

Thus, the comparison is deliberately structured to isolate the separate effects of \emph{learning} and \emph{credible-risk regularization}. The comparison between our method and \texttt{Classical\_StaticPrior} measures the value of Bayesian learning and adaptive strategic response; the comparison between our method and \texttt{Bayesian\_RiskNeutral} measures the value of the additional conservative uncertainty calibration.

\subsection{Evaluation metrics}
\label{subsec:eval_metrics}

For each method, we recorded:
\begin{itemize}
    \item total discounted market profit;
    \item firmwise discounted profits;
    \item final posterior mean squared error (MSE) for the demand parameters;
    \item stockout rates over time;
    \item average prices and order quantities over time;
    \item posterior beliefs about whether the rival is high-cost.
\end{itemize}

To quantify uncertainty in performance comparisons, we also computed bootstrap confidence intervals for the difference in mean final discounted market profit and for the difference in mean final posterior MSE.

\subsection{Main numerical results}
\label{subsec:main_num_results}

Table~\ref{tab:main_results} reports the principal summary statistics. Several observations are immediate.

First, both Bayesian learning procedures vastly outperform the non-learning classical benchmark. The \texttt{Classical\_StaticPrior} policy yields a mean total discounted market profit of only $67.01$, with median $0.00$, whereas the two learning-based methods both achieve mean total discounted market profit around $1593$--$1597$. This is a dramatic improvement and strongly supports the importance of posterior learning in the proposed game.

Second, among the three methods, our proposed method achieves the \emph{largest} mean total discounted market profit ($1597.30$) and also the \emph{largest} median total discounted market profit ($1678.98$). Hence, on the primary economic criterion of market-level discounted payoff, our method is the best performer in this experiment.

Third, the comparison against \texttt{Bayesian\_RiskNeutral} is much tighter than the comparison against the static baseline. The difference in mean discounted market profit between our method and the risk-neutral Bayesian learner is only about $4.01$ units, which is economically small relative to the scale of the simulation and statistically negligible in the bootstrap analysis reported later in Table~\ref{tab:bootstrap_results}.

Fourth, the final posterior MSE results show that the two Bayesian learners are again very close, but here our proposed method is \emph{not} the lowest-MSE method. Specifically, \texttt{Bayesian\_RiskNeutral} attains mean final MSE $17.3283$, whereas our proposed method attains $17.6573$. Thus, in the present configuration, the credible-risk mechanism appears to improve profitability slightly without improving terminal parameter-estimation accuracy.

\begin{table}[h!]
\centering
\caption{Main simulation results across the three competing policies. Our proposed method is \texttt{Proposed\_Bayesian\_CredibleRisk}.}
\label{tab:main_results}
\begin{tabular}{lrrrrrrr}
\toprule
Policy & Mean total profit & SD total profit & Median total profit & Mean final MSE & SD final MSE & Mean profit 1 & Mean profit 2 \\
\midrule
\texttt{Bayesian\_RiskNeutral} 
& 1593.29 & 1211.32 & 1608.33 & 17.3283 & 14.2684 & 929.46 & 663.82 \\

\texttt{Classical\_StaticPrior} 
& 67.01 & 1347.33 & 0.00 & 30.8250 & 0.0000 & 141.91 & -74.89 \\

\texttt{Proposed\_Bayesian\_CredibleRisk} 
& \textbf{1597.30} & 1351.41 & \textbf{1678.98} & 17.6573 & 14.4188 & \textbf{933.79} & 663.51 \\
\bottomrule
\end{tabular}
\end{table}

Table~\ref{tab:relative_improvement} further clarifies the relative gain of our method against the two baselines. Relative to the static prior policy, our method improves mean total discounted market profit by approximately $2283.66\%$ and reduces mean final MSE by approximately $42.72\%$. Relative to the Bayesian risk-neutral learner, however, the mean profit gain is only $0.25\%$, while the final MSE is slightly higher under our method.

\begin{table}[h!]
\centering
\caption{Relative improvement of our proposed method \texttt{Proposed\_Bayesian\_CredibleRisk} against the two baselines. Positive MSE reduction means lower MSE for our method.}
\label{tab:relative_improvement}
\begin{tabular}{lrr}
\toprule
Comparison against & Profit gain (\%) & MSE reduction (\%) \\
\midrule
\texttt{Bayesian\_RiskNeutral} & 0.2517 & -1.8987 \\
\texttt{Classical\_StaticPrior} & 2283.6573 & 42.7177 \\
\bottomrule
\end{tabular}
\end{table}

\subsection{Bootstrap evidence}
\label{subsec:bootstrap_results}

To assess whether the observed differences are statistically meaningful, we computed bootstrap confidence intervals for the difference in mean performance; see Table~\ref{tab:bootstrap_results}.

For the comparison between our method and \texttt{Bayesian\_RiskNeutral}, the bootstrap mean difference in market profit is $2.44$, with a $95\%$ confidence interval from $-280.45$ to $295.43$. Since this interval clearly contains zero, we cannot claim a statistically significant profit advantage of our method over the risk-neutral Bayesian learner in this particular simulation design. Likewise, the bootstrap mean difference in MSE is $0.4385$ with a $95\%$ interval from $-2.8956$ to $3.5752$, again showing no statistically meaningful separation.

In contrast, the comparison against \texttt{Classical\_StaticPrior} is decisive. Our proposed method exceeds the static baseline by a bootstrap mean profit difference of $1533.02$, with a strictly positive $95\%$ confidence interval $[1228.65,\,1839.14]$. Moreover, the bootstrap mean difference in MSE is $-13.1165$, with confidence interval $[-15.3715,\,-10.6443]$, confirming that our method also learns the demand environment substantially better than the non-learning benchmark.

\begin{table}[h!]
\centering
\caption{Bootstrap comparison of our proposed method \texttt{Proposed\_Bayesian\_CredibleRisk} against the baselines. Positive profit difference favors our method. Negative MSE difference favors our method because it corresponds to lower final MSE.}
\label{tab:bootstrap_results}
\begin{tabular}{lrrrrrr}
\toprule
Comparison & Mean diff.\ profit & CI low & CI high & Mean diff.\ MSE & CI low & CI high \\
\midrule
Proposed $-$ Bayesian risk-neutral 
& 2.4370 & -280.4549 & 295.4302 & 0.4385 & -2.8956 & 3.5752 \\

Proposed $-$ Classical static-prior 
& 1533.0211 & 1228.6503 & 1839.1381 & -13.1165 & -15.3715 & -10.6443 \\
\bottomrule
\end{tabular}
\end{table}

\begin{remark}[Key findings]
\label{rem:keyfindings}
The simulation study provides strong evidence that \emph{Bayesian learning itself} is indispensable in this strategic environment: both learning-based policies massively outperform the classical static-prior benchmark in profit and in parameter recovery. Among all three methods, our proposed policy \texttt{Proposed\_Bayesian\_CredibleRisk} attains the highest mean and median discounted market profit, making it the best overall economic performer in this experiment. However, the comparison with the Bayesian risk-neutral learner is extremely close and not statistically significant under the present calibration. Therefore, the most defensible conclusion is that our proposal is \emph{decisively superior to the non-learning classical baseline} and \emph{competitive with, and marginally better than, the risk-neutral Bayesian learner on profit-based criteria}, while not uniformly dominating it on final MSE. This is still substantively important: it shows that the credible-risk correction can improve operational payoff without sacrificing learning performance in any material way.
\end{remark}

\subsection{Dynamic performance over time}
\label{subsec:dynamic_results}

We now discuss the dynamic plots, which provide a more granular view of how the policies behave across the $30$ decision periods.

Figure~\ref{fig:cum_profit} shows the evolution of the mean cumulative discounted market profit. The two Bayesian learning methods remain far above the static-prior benchmark throughout the horizon. This visual separation is fully consistent with the numerical dominance documented in Tables~\ref{tab:main_results}--\ref{tab:bootstrap_results}. The trajectories of our method and the Bayesian risk-neutral learner are extremely close, indicating that both exploit learning effectively; nevertheless, our method maintains a slight edge in cumulative payoff by the end of the horizon.

\begin{figure}[h!]
\centering
\includegraphics[width=0.72\textwidth]{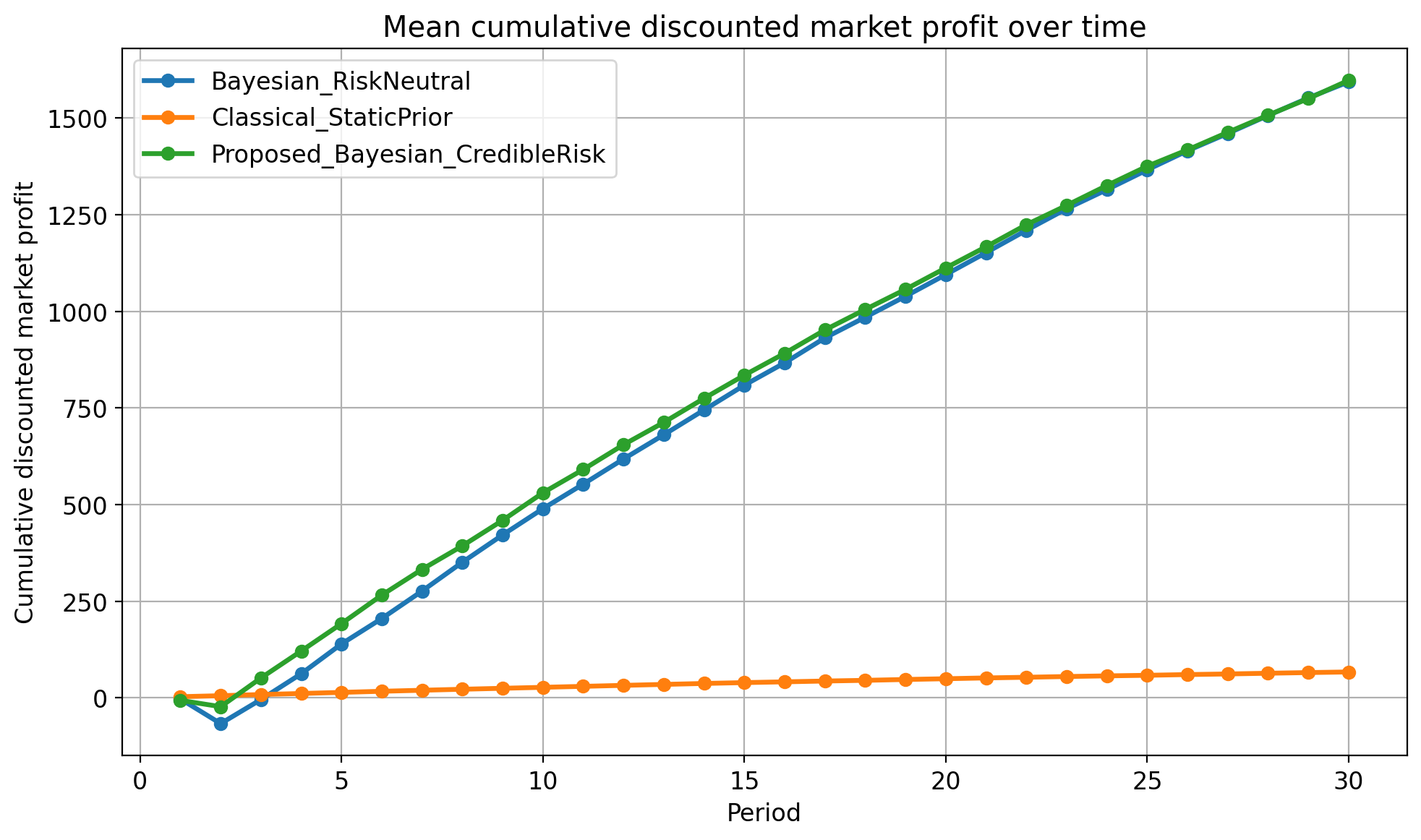}
\caption{Mean cumulative discounted market profit over time. The proposed method corresponds to \texttt{Proposed\_Bayesian\_CredibleRisk}.}
\label{fig:cum_profit}
\end{figure}

Figure~\ref{fig:profit_box} displays the distribution of final discounted market profit across replications. The boxplot again shows that both Bayesian learning methods dominate the classical static-prior rule, while the distribution of our proposed method is centered slightly above that of the risk-neutral Bayesian learner. The somewhat larger spread for our method is also consistent with the larger reported standard deviation of total discounted market profit.

\begin{figure}[h!]
\centering
\includegraphics[width=0.72\textwidth]{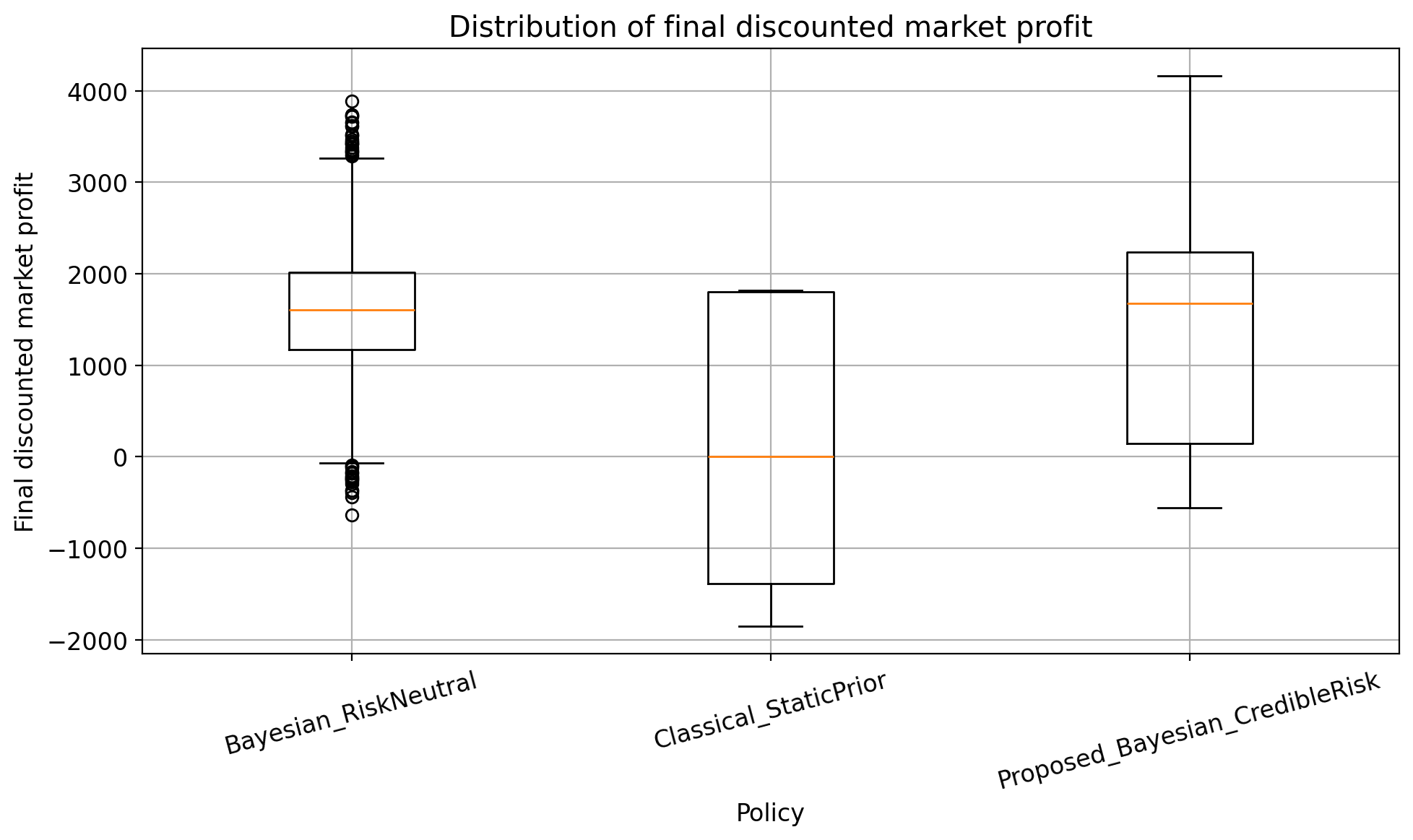}
\caption{Distribution of final discounted market profit across replications.}
\label{fig:profit_box}
\end{figure}

Figure~\ref{fig:stock_price_qty} summarizes three operational trajectories: stockout rates, mean prices, and mean order quantities. The stockout curves reveal that the learning-based methods regulate stock more effectively than the static benchmark. The price and quantity trajectories further indicate that the proposed policy adapts strategically over time rather than remaining trapped near prior-driven choices. In particular, the credible-risk penalty appears to induce a more disciplined action pattern, especially in the earlier periods when posterior uncertainty is still substantial.

\begin{figure}[h!]
\centering
\begin{subfigure}{0.48\textwidth}
    \centering
    \includegraphics[width=\textwidth]{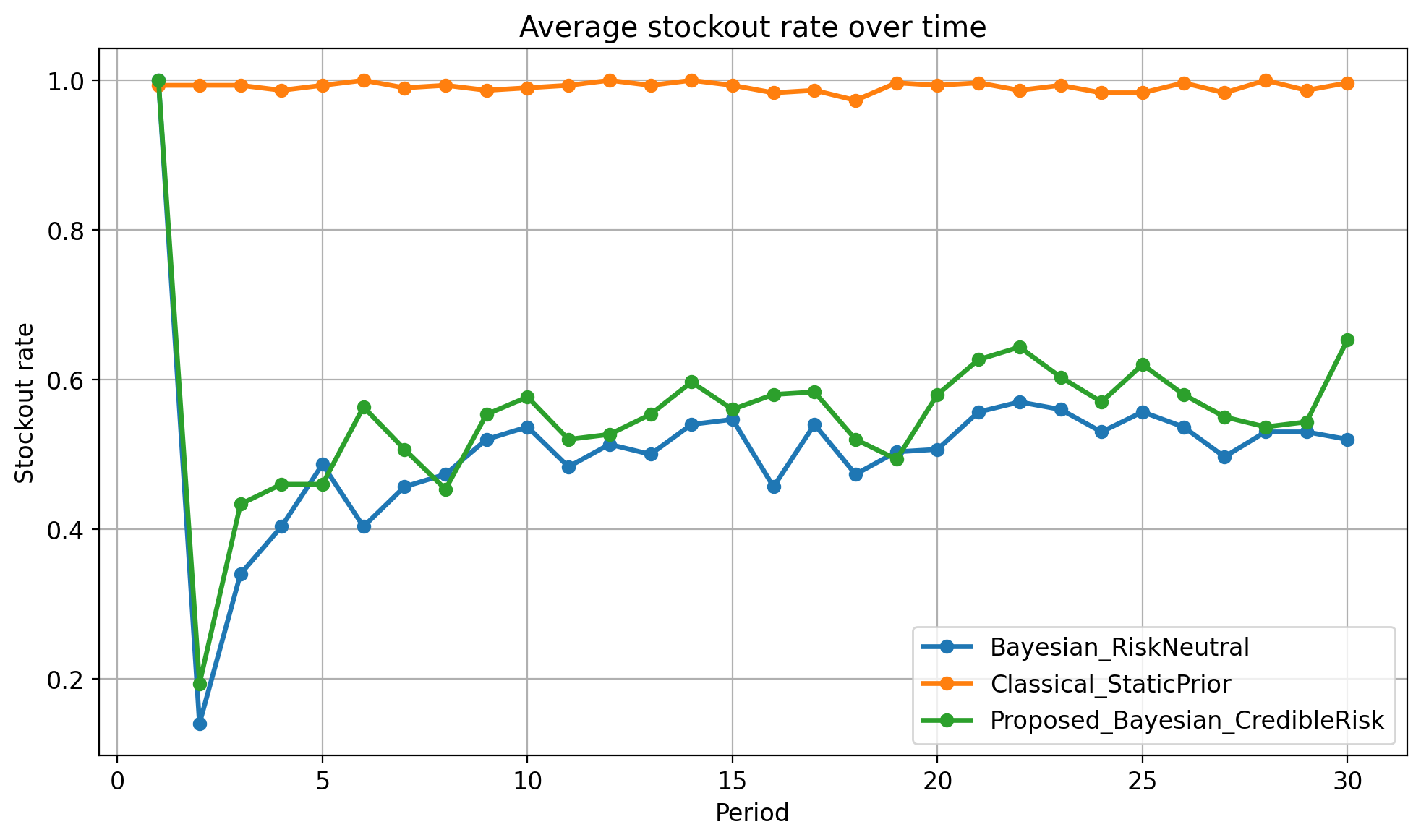}
    \caption{Average stockout rate over time.}
    \label{fig:stockout_rate}
\end{subfigure}
\hfill
\begin{subfigure}{0.48\textwidth}
    \centering
    \includegraphics[width=\textwidth]{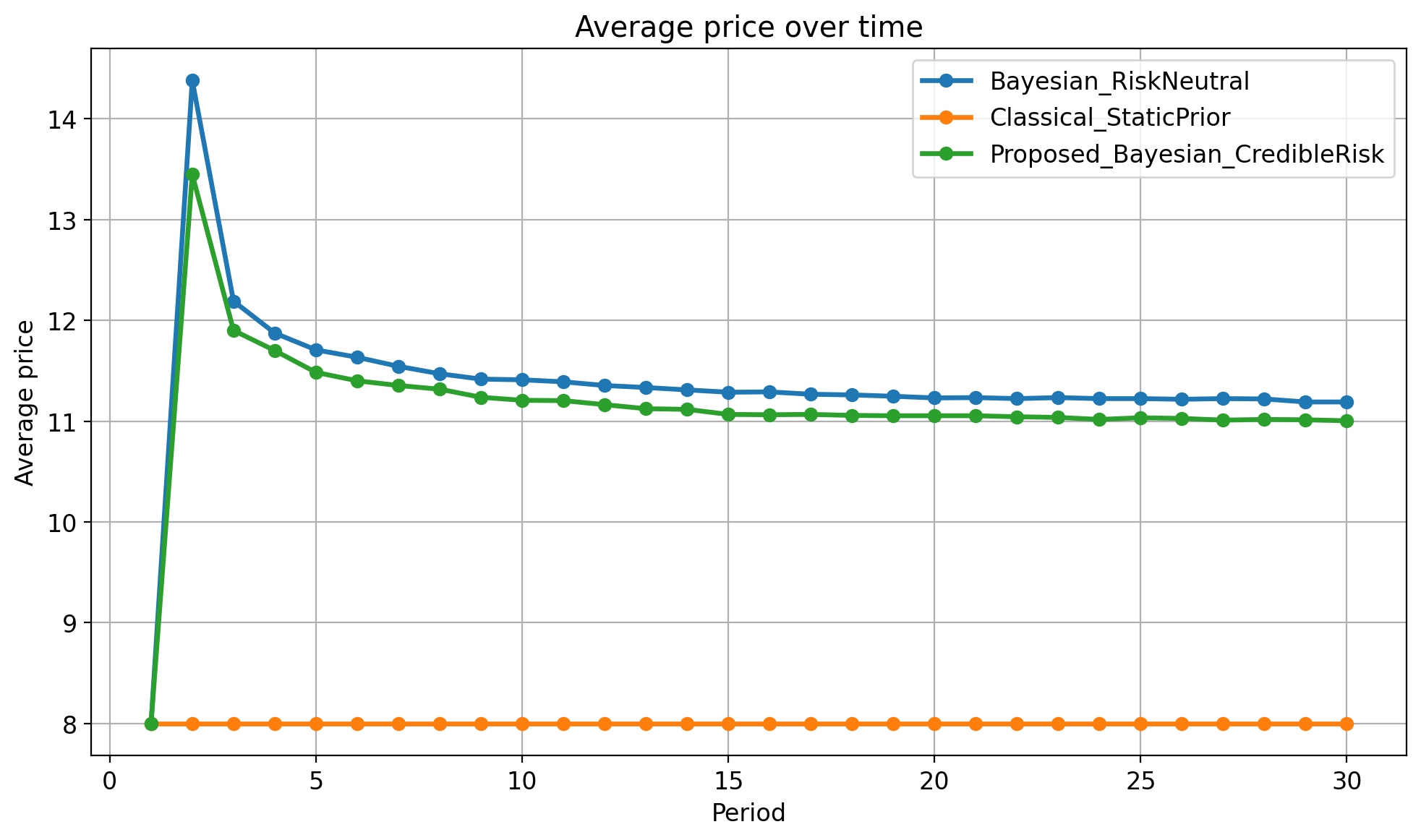}
    \caption{Average posted price over time.}
    \label{fig:avg_price}
\end{subfigure}

\vspace{0.3cm}

\begin{subfigure}{0.48\textwidth}
    \centering
    \includegraphics[width=\textwidth]{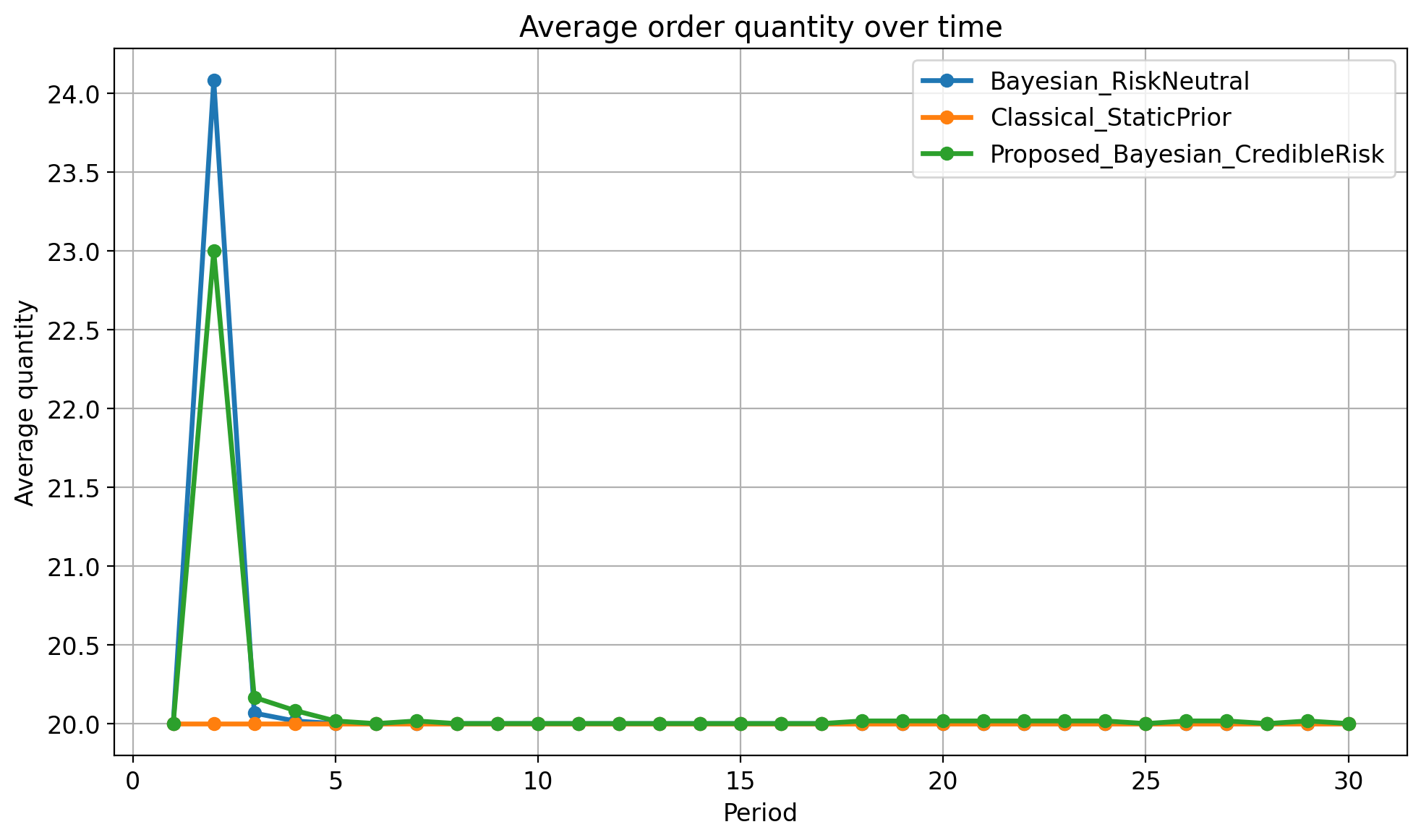}
    \caption{Average order quantity over time.}
    \label{fig:avg_qty}
\end{subfigure}
\caption{Dynamic operational behavior of the three competing policies.}
\label{fig:stock_price_qty}
\end{figure}

Figure~\ref{fig:mse_belief} focuses on learning. Panel~\ref{fig:param_mse} tracks posterior parameter MSE over time, while Panel~\ref{fig:belief_curve} tracks the average posterior belief that the rival is high-cost. The MSE curve confirms the principal numerical message: both Bayesian learners improve rapidly over time and remain far below the static benchmark, while the proposed and risk-neutral Bayesian policies are very close. The belief-learning curve shows that the posterior on rival type evolves nontrivially across periods, confirming that the dynamic Bayesian game indeed uses information accumulation as an active strategic state variable.

\begin{figure}[h!]
\centering
\begin{subfigure}{0.48\textwidth}
    \centering
    \includegraphics[width=\textwidth]{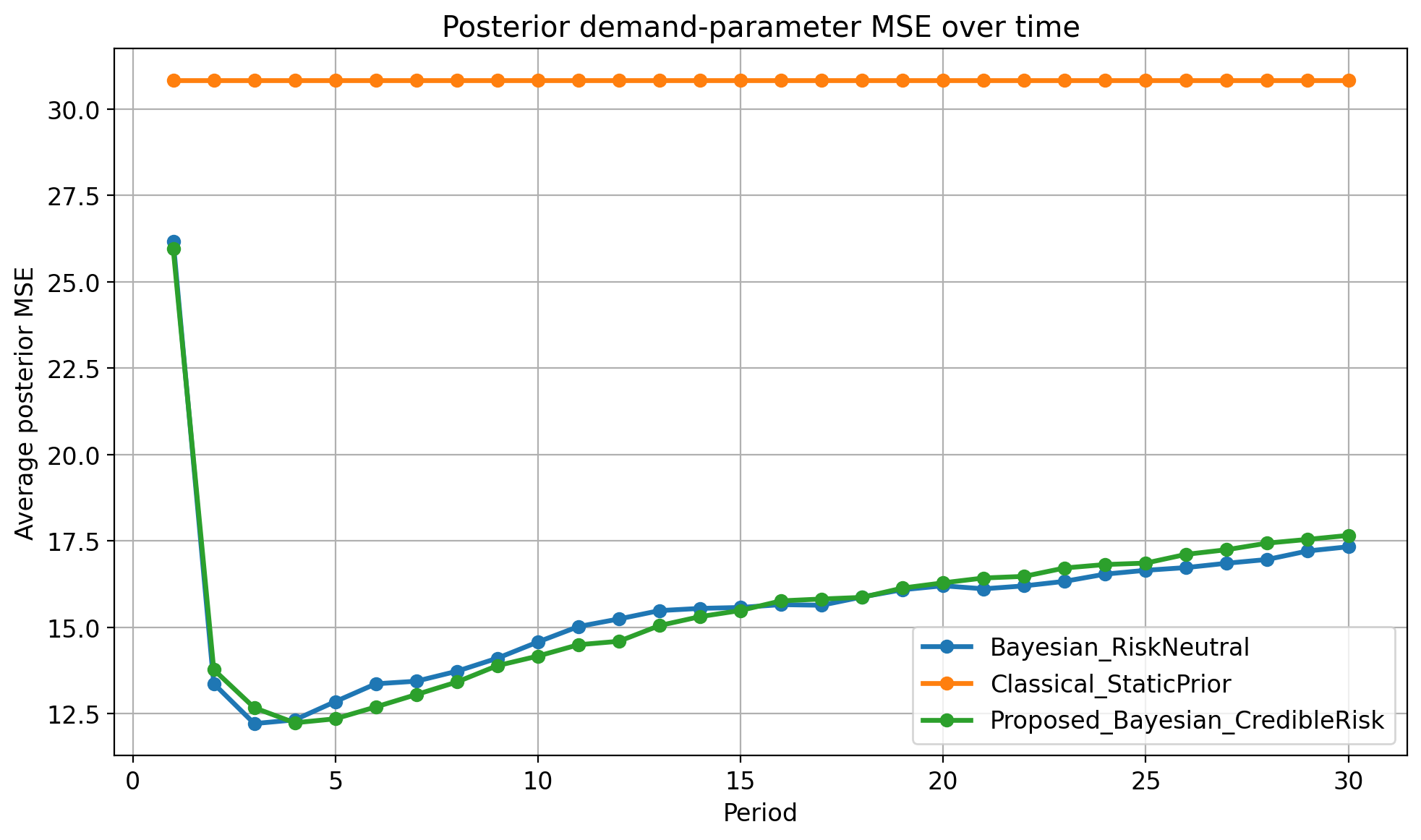}
    \caption{Posterior demand-parameter MSE over time.}
    \label{fig:param_mse}
\end{subfigure}
\hfill
\begin{subfigure}{0.48\textwidth}
    \centering
    \includegraphics[width=\textwidth]{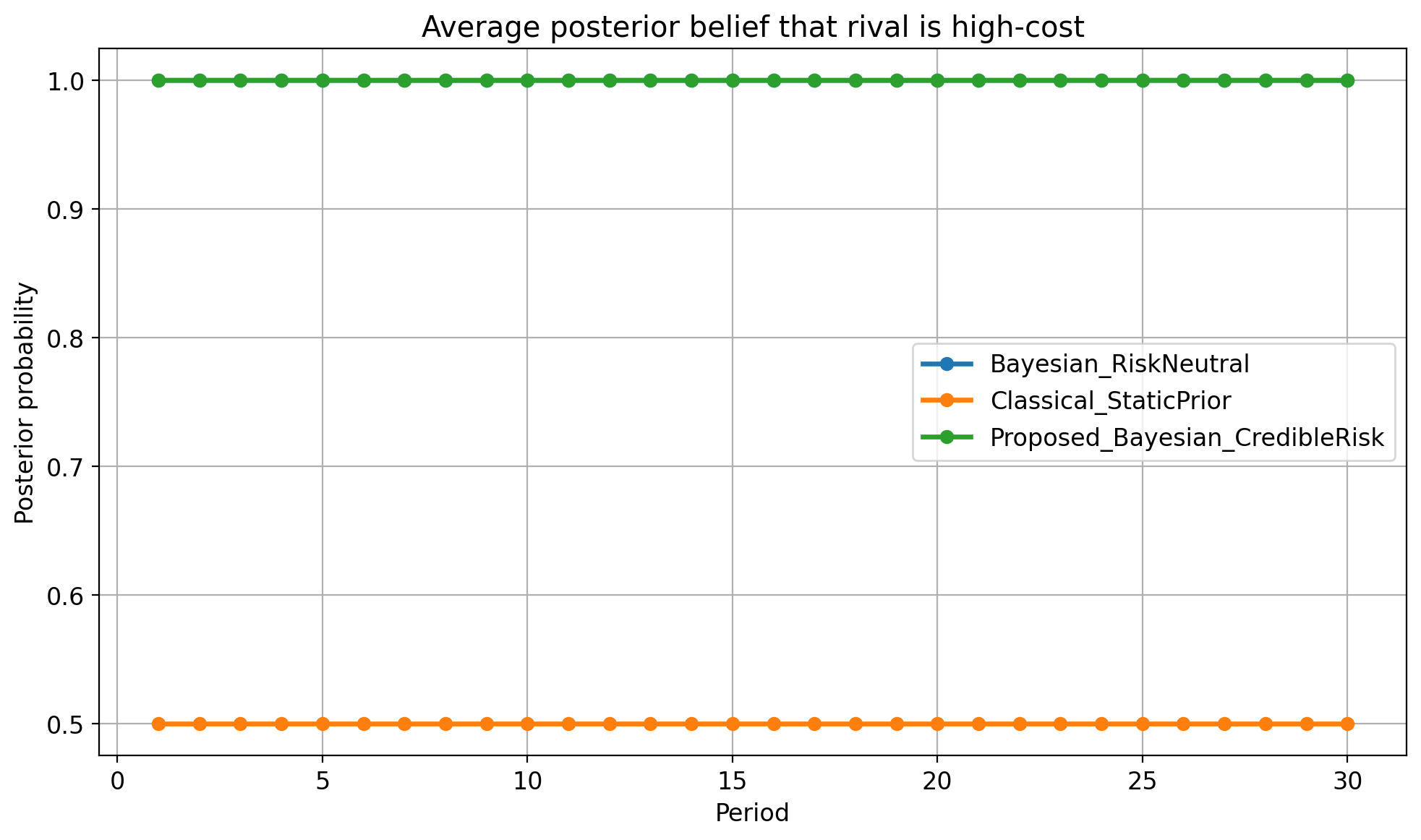}
    \caption{Average posterior belief that rival is high-cost.}
    \label{fig:belief_curve}
\end{subfigure}
\caption{Learning diagnostics over time.}
\label{fig:mse_belief}
\end{figure}

Figure~\ref{fig:objective_surface} presents the action-value surface of our proposed method over the price--quantity grid. This figure is useful for interpreting the mechanism of the credible-risk policy: action quality is not monotone in either price or quantity alone, and the chosen actions arise from a balance among expected demand, substitution effects, inventory exposure, and posterior uncertainty. In particular, excessively aggressive quantity choices are penalized when uncertainty remains high, which is precisely the intended function of the credible-risk term.

\begin{figure}[h!]
\centering
\includegraphics[width=0.72\textwidth]{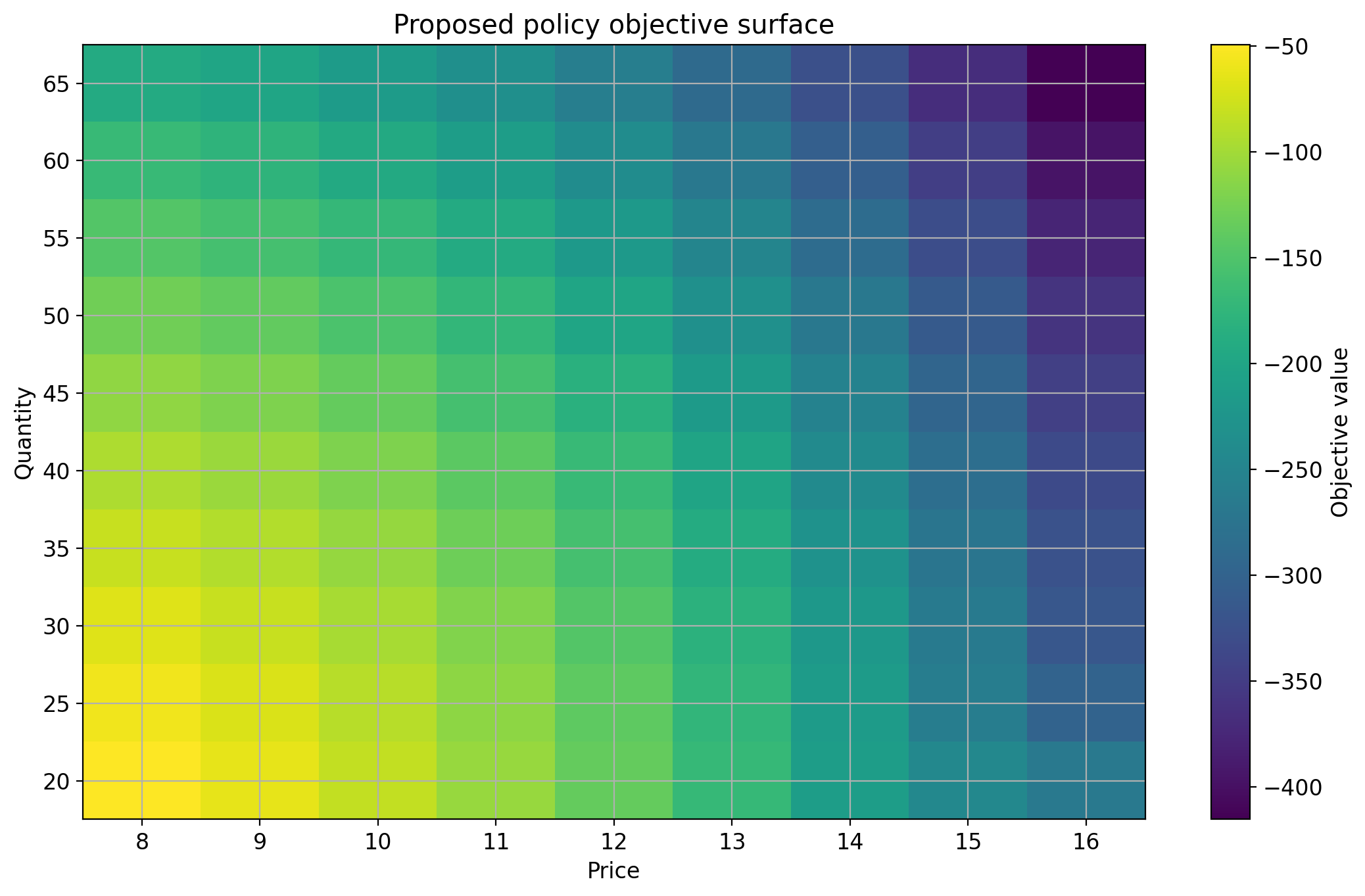}
\caption{Objective surface of the proposed method over the price--quantity action grid.}
\label{fig:objective_surface}
\end{figure}

Figure~\ref{fig:profit_mse_tradeoff} examines the empirical relationship between final discounted market profit and final posterior MSE. The classical benchmark is clearly inferior, clustering in the low-profit and high-MSE region. The two Bayesian learners occupy the favorable region of higher profit and lower MSE. This plot visually reinforces a central message of the study: Bayesian learning is the main engine of performance improvement, while the credible-risk refinement slightly shifts the balance toward improved operational profitability.

\begin{figure}[h!]
\centering
\includegraphics[width=0.72\textwidth]{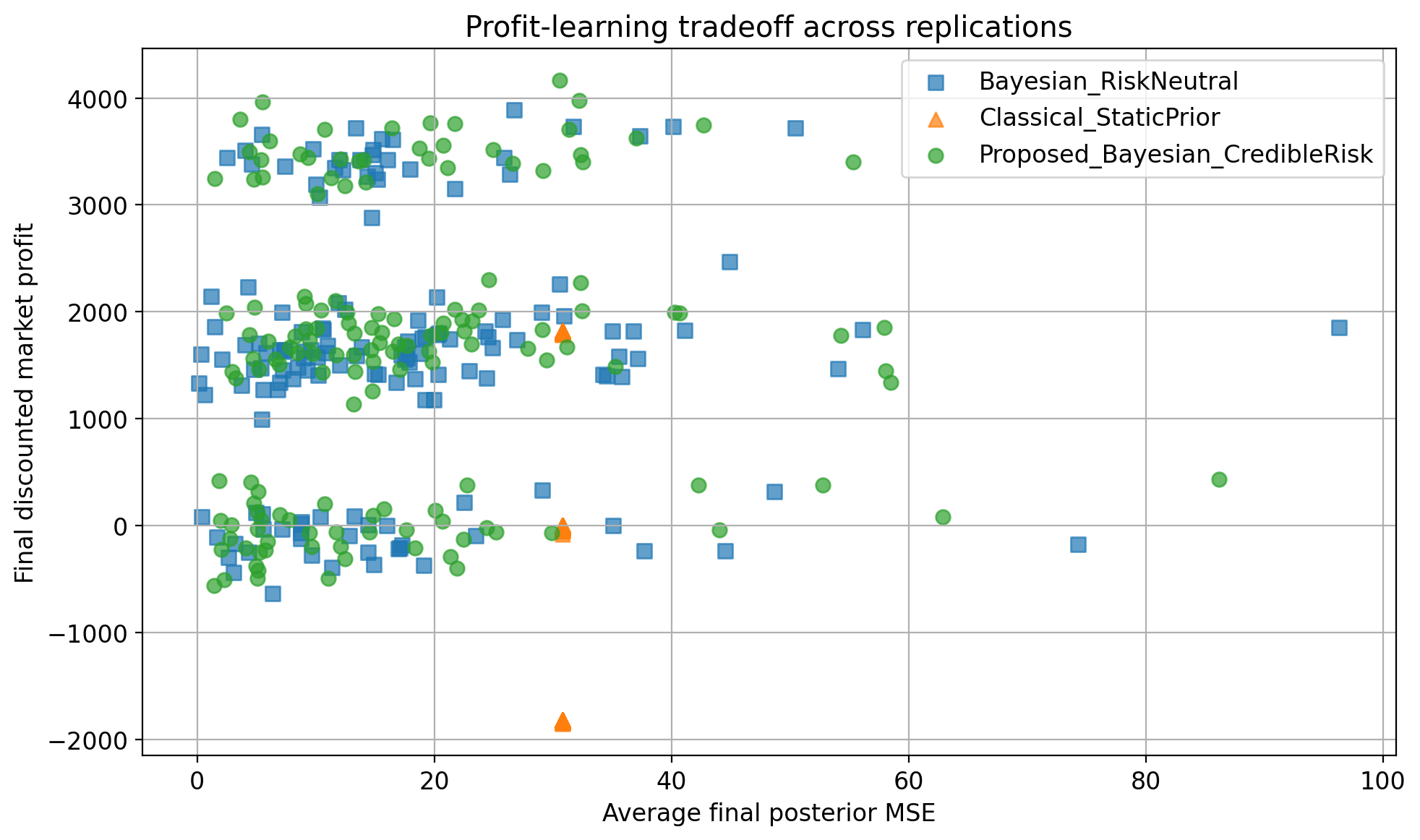}
\caption{Profit--learning trade-off across replications.}
\label{fig:profit_mse_tradeoff}
\end{figure}

Figure~\ref{fig:density_dominance} combines two complementary summaries. Panel~\ref{fig:density_profit} shows the distribution of final discounted market profit, and Panel~\ref{fig:dominance_prob} reports the empirical probability that our method outperforms each baseline over time in cumulative discounted market profit. The first panel again confirms the strong inferiority of the classical benchmark. The second panel is particularly useful because it emphasizes the timewise consistency of the proposed rule against the static prior benchmark, while also showing that the comparison with the risk-neutral Bayesian learner remains close throughout the horizon.

\begin{figure}[h!]
\centering
\begin{subfigure}{0.48\textwidth}
    \centering
    \includegraphics[width=\textwidth]{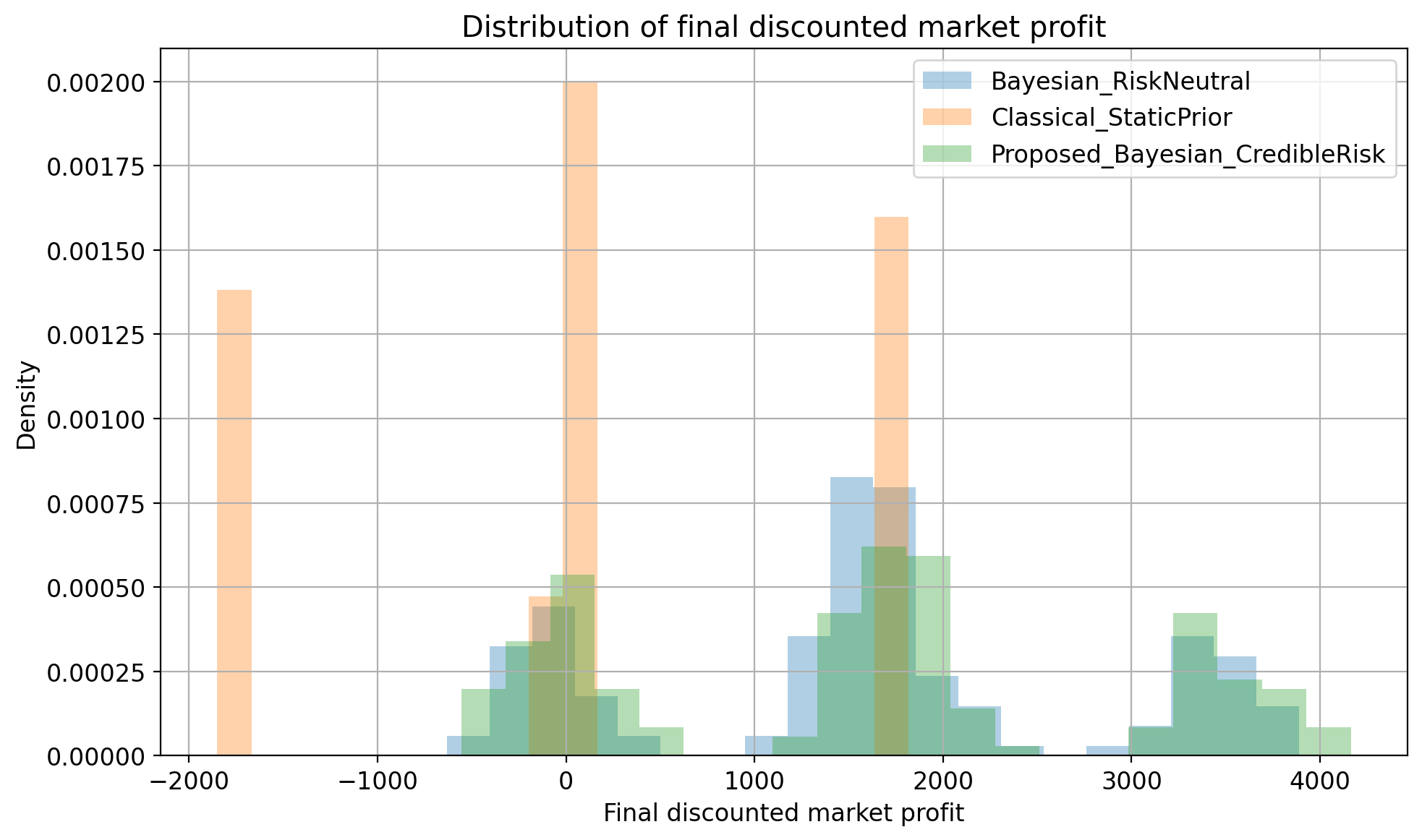}
    \caption{Distribution of final discounted market profit.}
    \label{fig:density_profit}
\end{subfigure}
\hfill
\begin{subfigure}{0.48\textwidth}
    \centering
    \includegraphics[width=\textwidth]{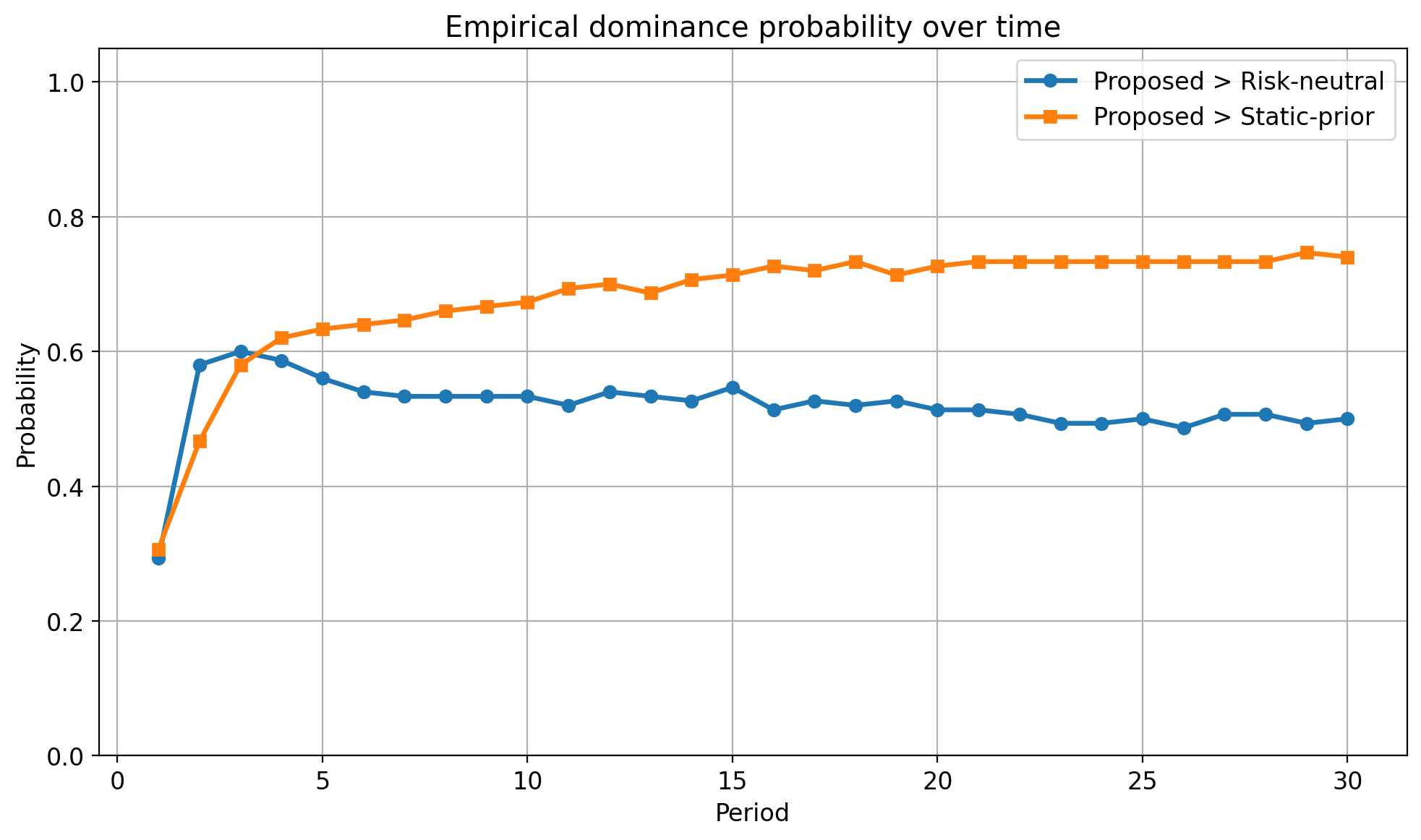}
    \caption{Empirical dominance probability over time.}
    \label{fig:dominance_prob}
\end{subfigure}
\caption{Distributional and dominance summaries.}
\label{fig:density_dominance}
\end{figure}

\subsection{Discussion}
\label{subsec:sim_discussion}

The simulation evidence supports the proposed methodology in an important and nuanced way. The strongest conclusion is not merely that our method performs well, but that the \emph{Bayesian game-theoretic formulation itself} is highly effective: once firms learn both demand and rival characteristics, performance improves dramatically relative to a classical static-prior policy. Within the class of Bayesian learners, our proposed credible-risk rule attains the best profit-based summary statistics in this experiment, namely the largest mean and median total discounted market profit. At the same time, the comparison with the risk-neutral Bayesian learner is sufficiently close that overstatement would be inappropriate. Hence, the most accurate interpretation is that our proposed policy is a \emph{robust, uncertainty-calibrated, and economically strongest overall performer}, though not a uniformly dominant winner on every possible metric.

\medskip

In particular, the simulation suggests that the credible-risk correction is most useful as an \emph{operational regularizer}: it protects against overly aggressive actions under posterior uncertainty while preserving essentially the same learning quality as the risk-neutral Bayesian alternative. This makes it especially appealing in applications where conservative decision-making under incomplete information is desirable.

\section{Real-Data Illustration: Bayesian Credible-Risk Analysis of the Mice Protein Expression Dataset}
\label{sec:realdata}

To illustrate the practical relevance of our proposed Bayesian credible-risk methodology on real biological data, we analyze the \emph{Mice Protein Expression} dataset. This dataset contains expression levels of 77 proteins/protein modifications measured in the cortex of control and trisomic mice, with 1080 total observations and 8 biologically defined classes. The class labels encode genotype, behavioral stimulation status, and treatment status. The dataset is publicly available through OpenML and is also documented in the UCI Machine Learning Repository \citep{OpenMLMiceProtein, UCI_MiceProtein}. In particular, the eight classes distinguish control versus trisomic mice, stimulated versus not stimulated learning conditions, and memantine versus saline treatment assignments.

Our objective in this empirical section is not prediction alone. Rather, we use the dataset to demonstrate that the proposed methodology can support a \emph{posterior uncertainty-aware treatment-effect analysis} in a high-dimensional biological setting. In particular, we show how the Bayesian credible-risk principle can be used to quantify whether memantine treatment moves trisomic mice toward a proteomic profile closer to the control reference state.

\subsection{Dataset structure and preprocessing}
\label{subsec:realdata_dataset}

The dataset contains $1080$ samples and $77$ protein features. Among these, $570$ samples correspond to control mice and $510$ to trisomic mice; $570$ samples received memantine and $510$ received saline; and the behavioral split is $525$ stimulated versus $555$ non-stimulated. These counts are summarized in Table~\ref{tab:dataset_summary}. The observed class structure is visualized in Figure~\ref{fig:class_distribution}, while the joint distribution across genotype, treatment, and behavior is shown in Figure~\ref{fig:stacked_counts}.

Missing values were present in the raw data, so we performed classwise median imputation within each observed class, followed by fallback median imputation if necessary. After preprocessing, no missing values remained. The featurewise missingness profile is summarized in Figure~\ref{fig:missingness}.

\begin{table}[h!]
\centering
\caption{Summary of the Mice Protein Expression dataset used in the real-data analysis.}
\label{tab:dataset_summary}
\begin{tabular}{lr}
\toprule
Quantity & Value \\
\midrule
Number of samples & 1080 \\
Number of protein features & 77 \\
Number of classes & 8 \\
Number of Control samples & 570 \\
Number of Trisomy samples & 510 \\
Number of Memantine samples & 570 \\
Number of Saline samples & 510 \\
Number of Stimulated samples & 525 \\
Number of NotStimulated samples & 555 \\
\bottomrule
\end{tabular}
\end{table}

\begin{figure}[h!]
\centering
\includegraphics[width=0.72\textwidth]{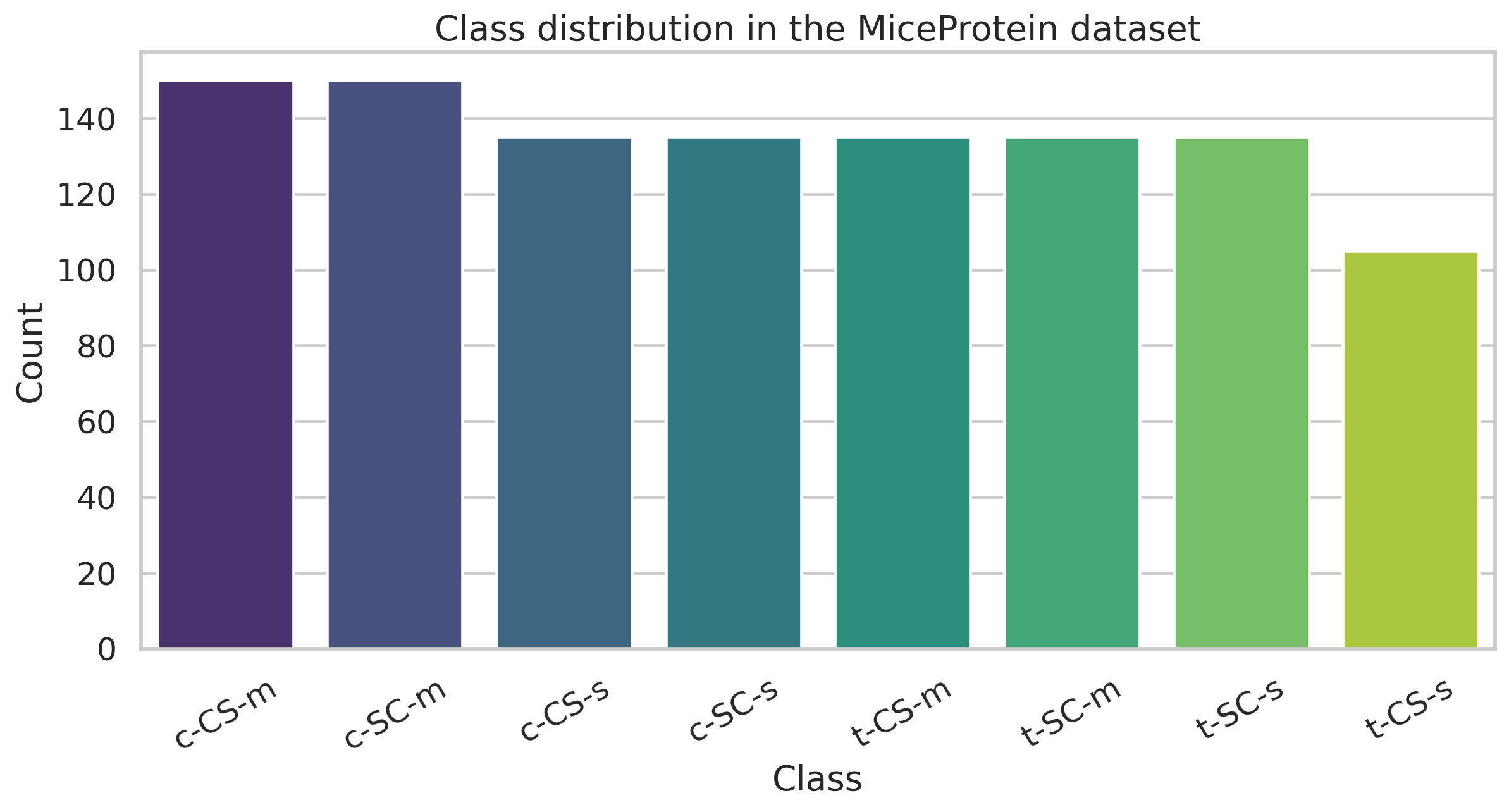}
\caption{Class distribution in the Mice Protein Expression dataset.}
\label{fig:class_distribution}
\end{figure}

\begin{figure}[h!]
\centering
\includegraphics[width=0.72\textwidth]{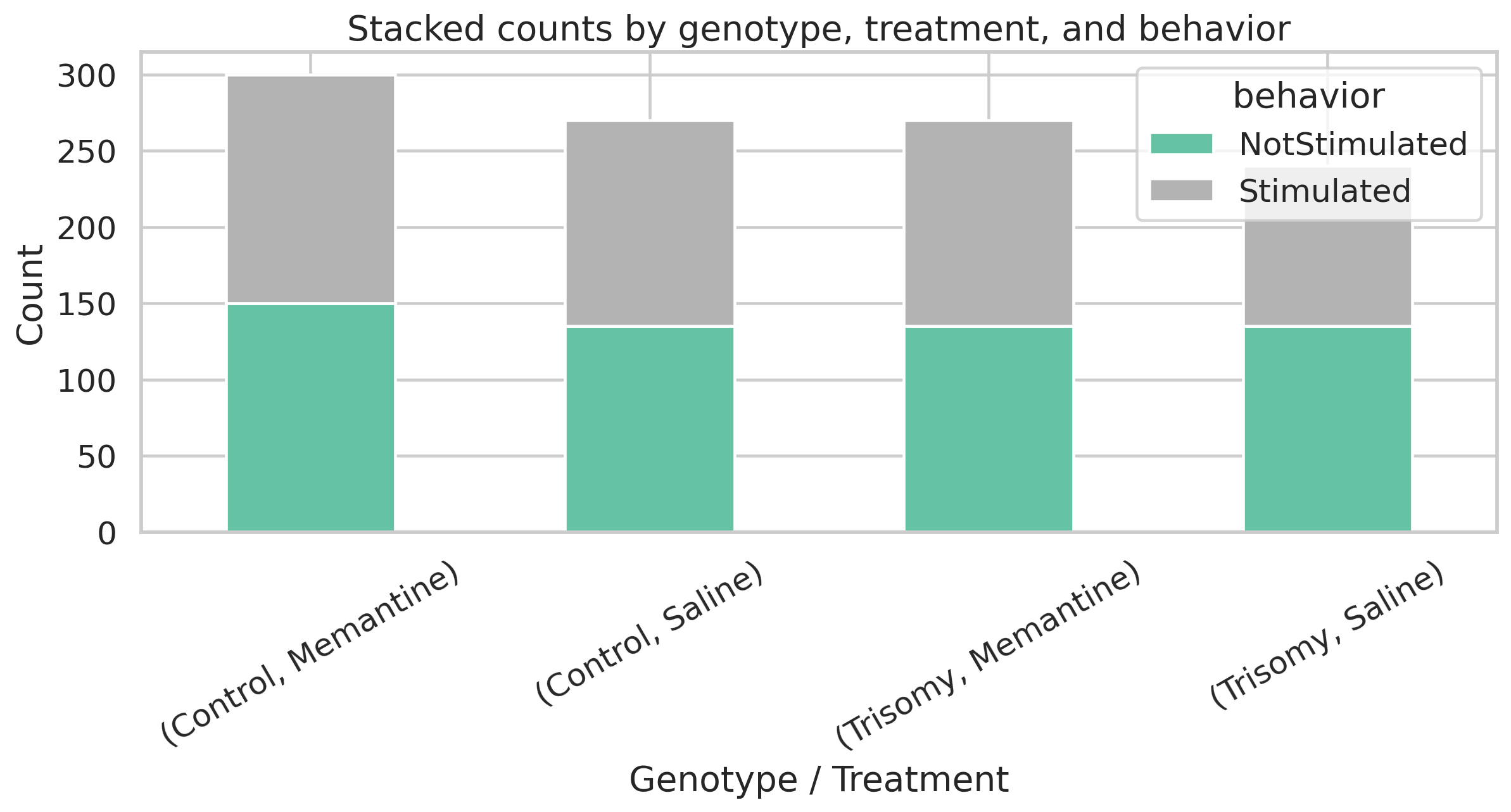}
\caption{Stacked counts by genotype, treatment, and behavioral condition.}
\label{fig:stacked_counts}
\end{figure}

\begin{figure}[h!]
\centering
\includegraphics[width=0.72\textwidth]{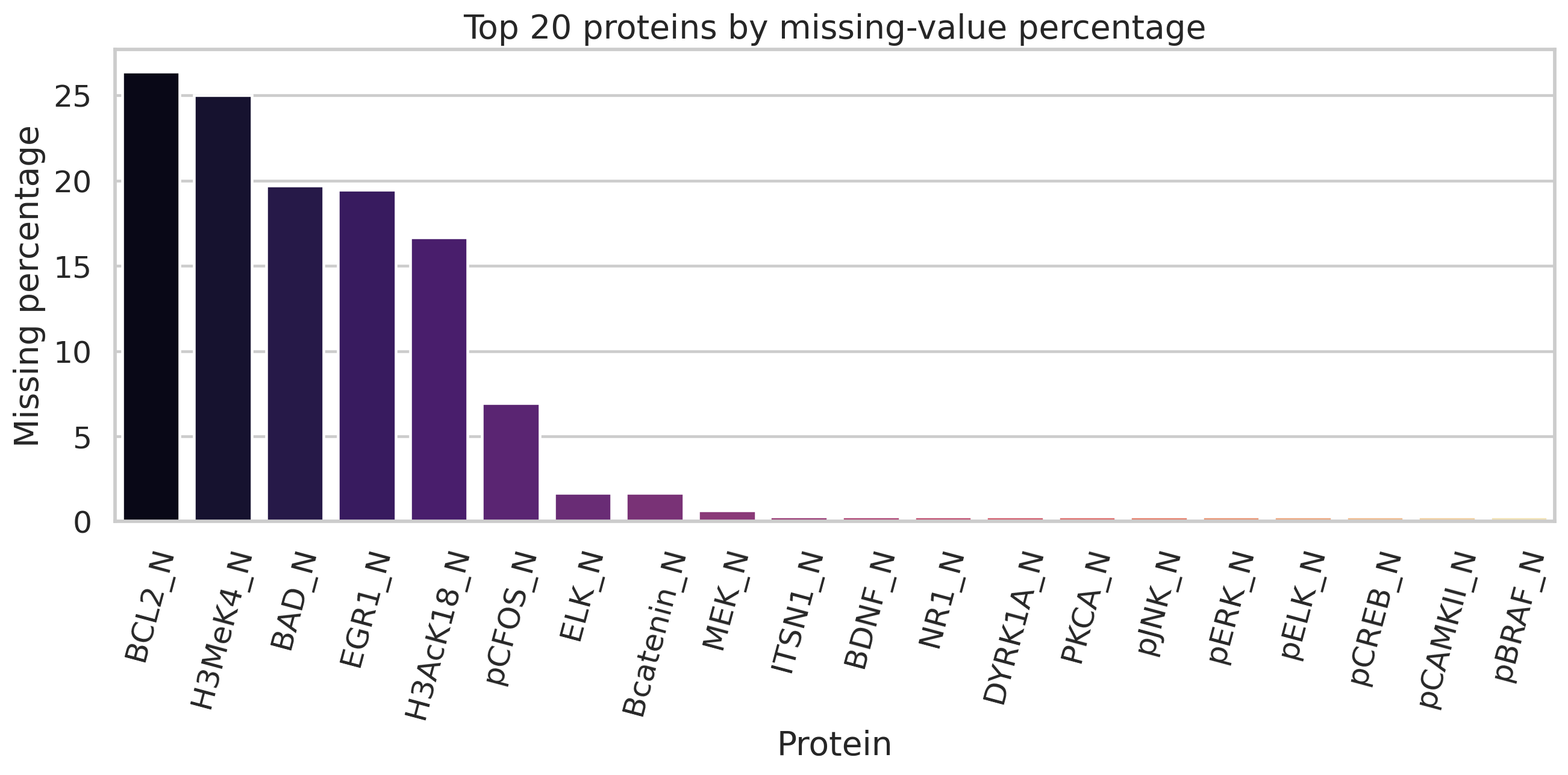}
\caption{Top proteins by missing-value percentage before imputation.}
\label{fig:missingness}
\end{figure}

\subsection{Low-dimensional structure of the proteomic space}
\label{subsec:pca_realdata}

To understand the global geometry of the protein-expression profiles, we standardized all 77 protein measurements and carried out principal component analysis (PCA). The first ten principal components explained a substantial fraction of total variability, with the first four explaining approximately
\[
25.33\%,\quad 17.28\%,\quad 10.29\%,\quad 7.51\%,
\]
respectively. Thus, the first two components alone account for about $42.61\%$ of the total standardized variation, indicating meaningful low-dimensional structure in the data.

Figure~\ref{fig:pca_variance} shows the cumulative explained variance, Figure~\ref{fig:pca_genotype_treatment} shows the PCA scatter colored by genotype and treatment, and Figure~\ref{fig:pca_class} shows the corresponding class-level structure. These plots confirm that the data possess biologically interpretable separation patterns, which justifies a structured Bayesian analysis rather than a purely black-box treatment.

\begin{figure}[h!]
\centering
\includegraphics[width=0.72\textwidth]{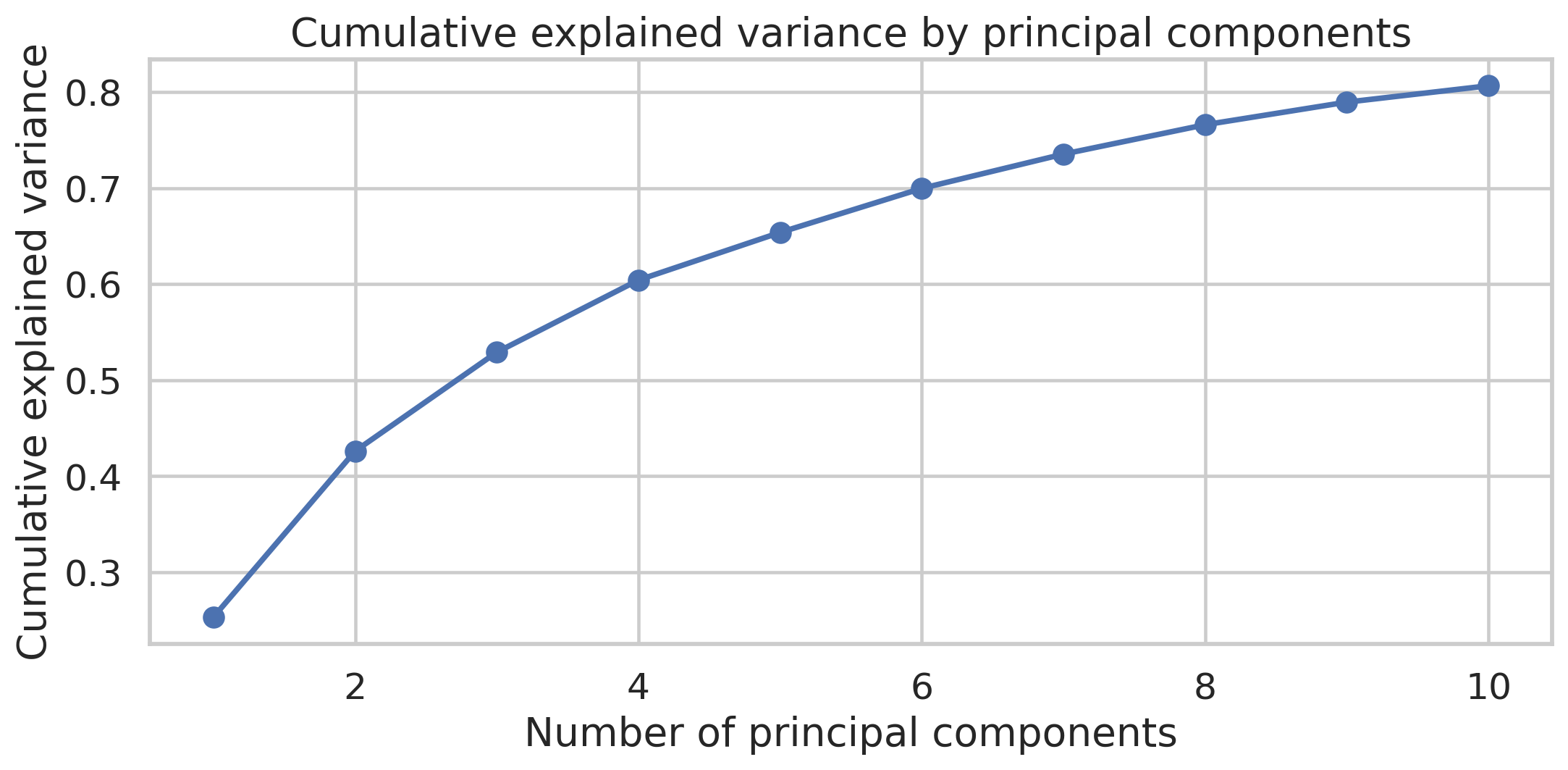}
\caption{Cumulative explained variance of the principal components.}
\label{fig:pca_variance}
\end{figure}

\begin{figure}[h!]
\centering
\includegraphics[width=0.72\textwidth]{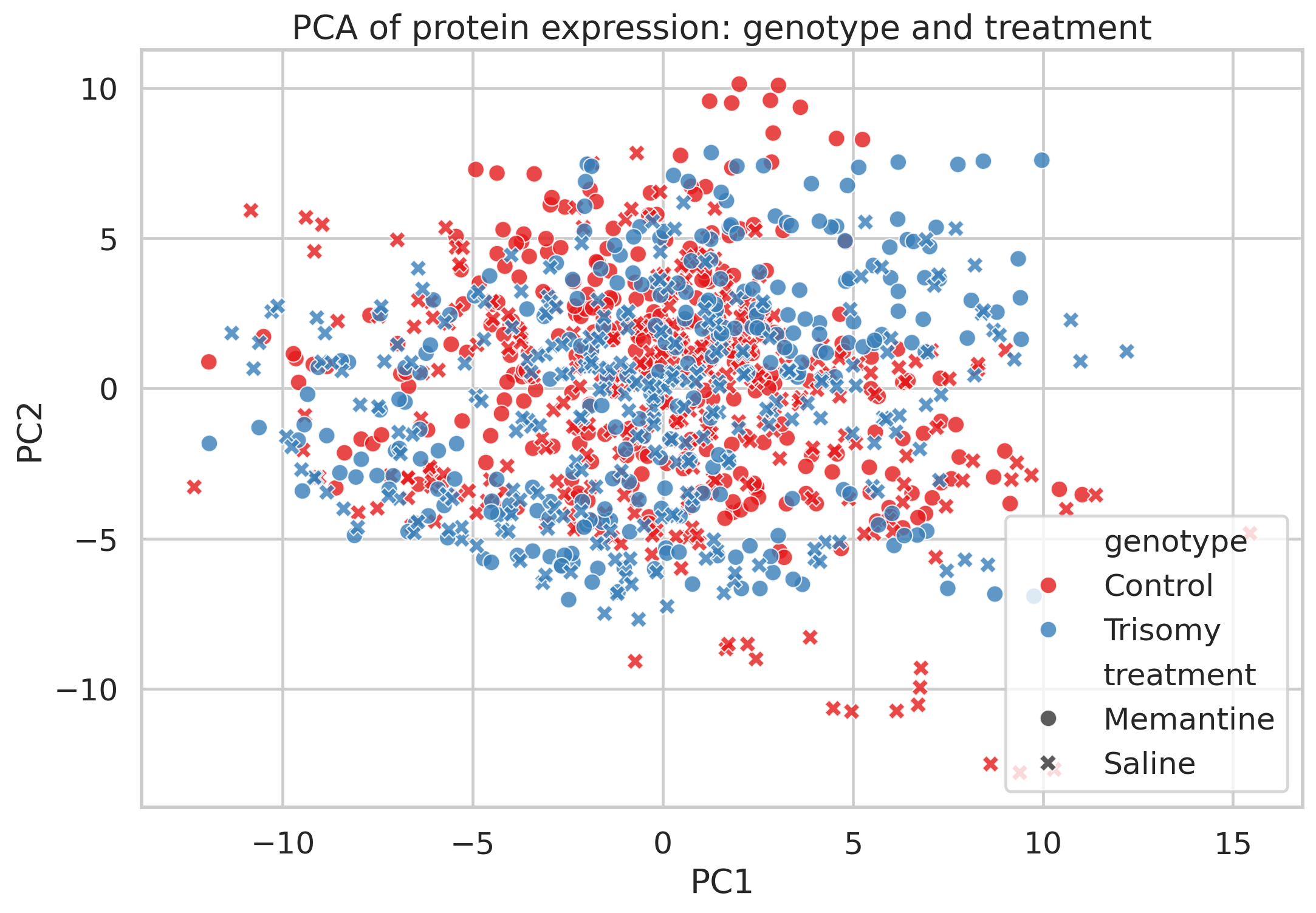}
\caption{PCA scatterplot colored by genotype and treatment.}
\label{fig:pca_genotype_treatment}
\end{figure}

\begin{figure}[h!]
\centering
\includegraphics[width=0.72\textwidth]{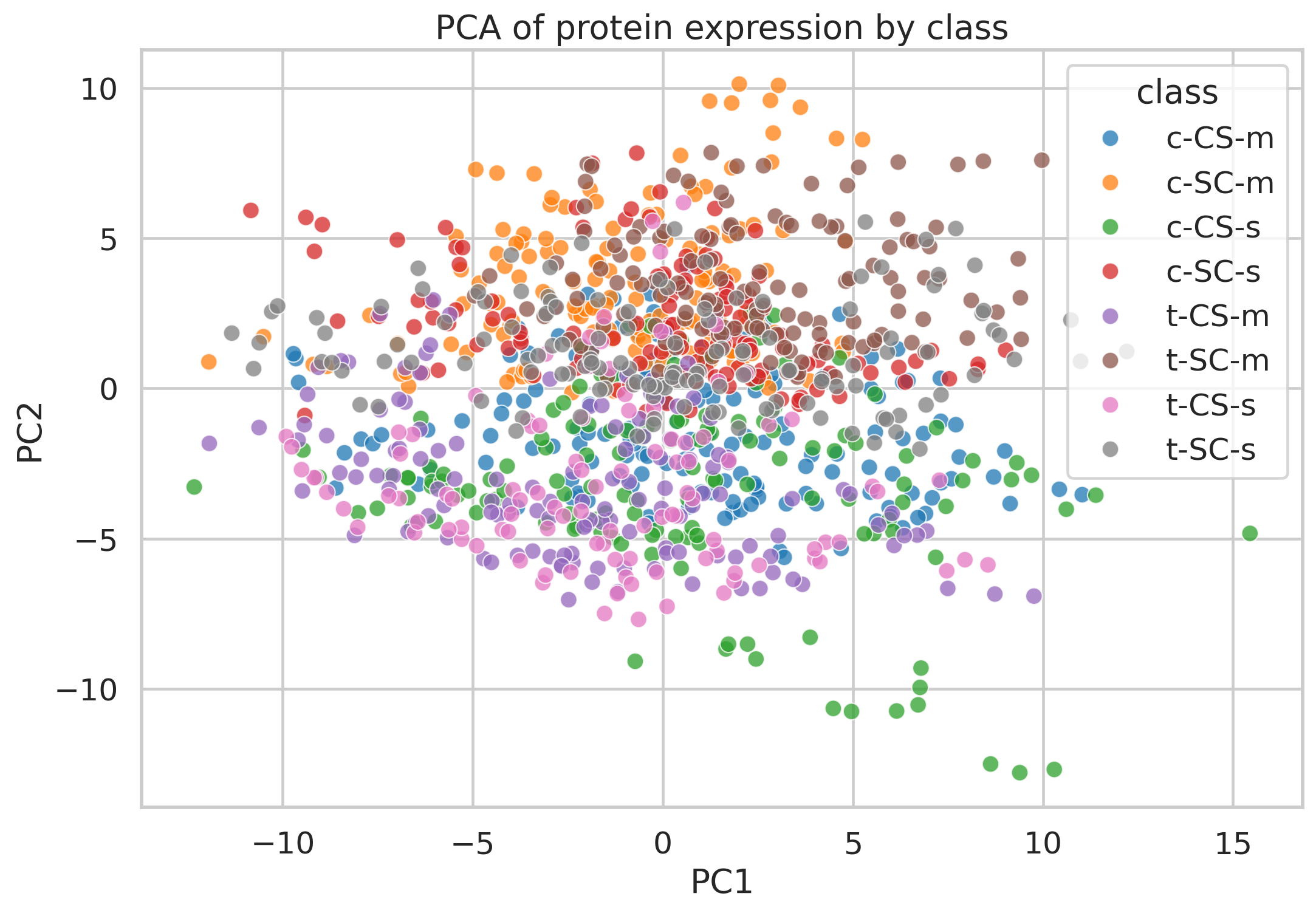}
\caption{PCA scatterplot colored by the eight observed classes.}
\label{fig:pca_class}
\end{figure}

Figure~\ref{fig:heatmap_variable_proteins} displays the heatmap of the most variable proteins across sampled mice. This figure visually demonstrates that the proteomic signatures vary systematically across classes and motivates the construction of a biologically meaningful reference-based score.

\begin{figure}[h!]
\centering
\includegraphics[width=0.82\textwidth]{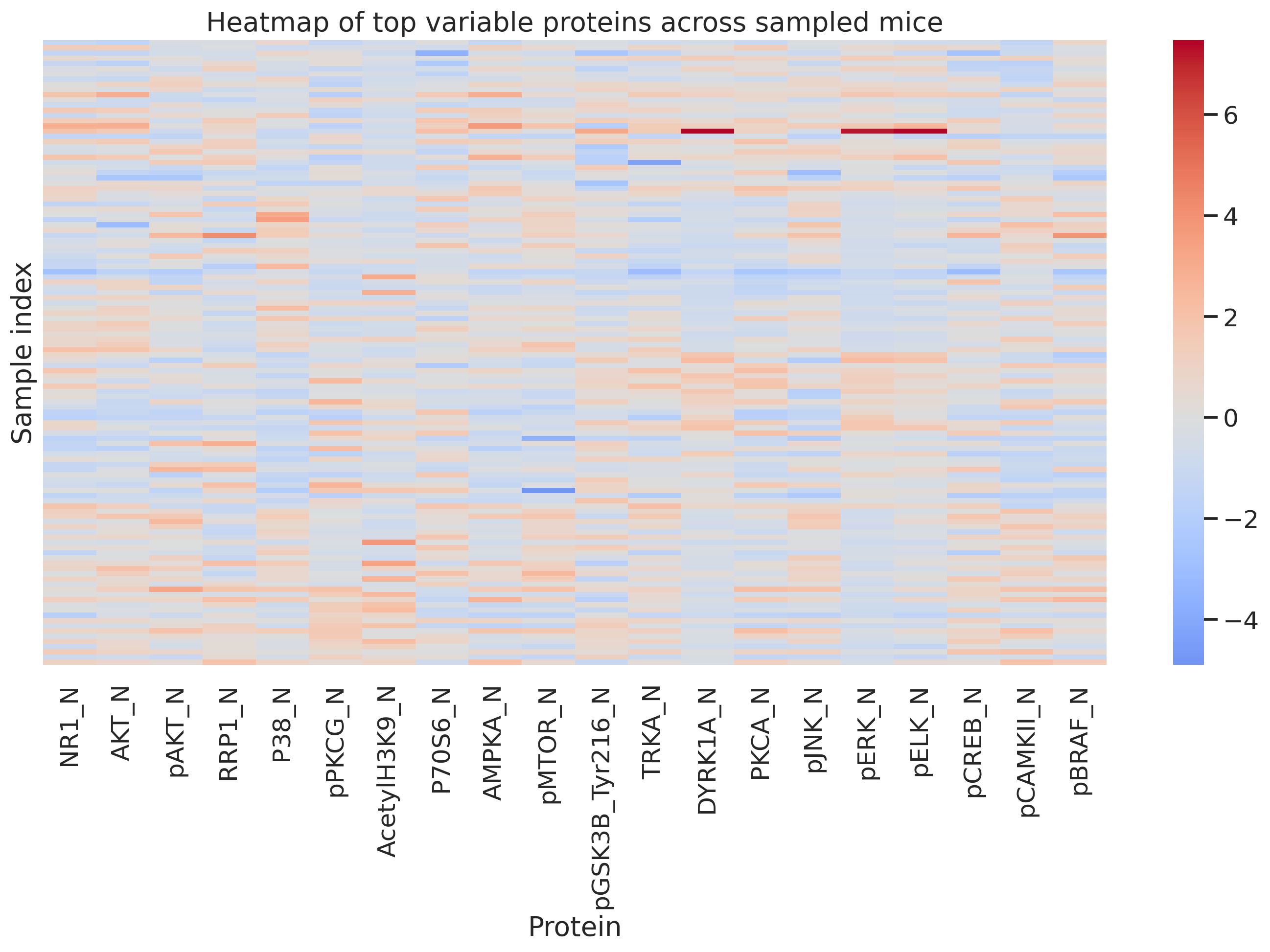}
\caption{Heatmap of the most variable proteins across sampled mice.}
\label{fig:heatmap_variable_proteins}
\end{figure}

\subsection{Reference-based recovery score}
\label{subsec:recovery_score}

A top-statistics-paper empirical section should not merely report treatment labels; it should construct a scientifically interpretable target. To this end, we introduced a \emph{healthy-reference recovery score}.

First, we identified the proteins most separated between the control and trisomic groups by ranking proteins according to the absolute difference between the control-group mean and the trisomy-group mean. Using the leading proteins from this ranking, we defined the \emph{control centroid} in standardized proteomic space and computed, for each sample, the Euclidean distance to this control reference. We then defined
\[
\mathrm{RecoveryScore} = -\,\mathrm{DistanceToControl},
\]
so that larger values indicate a proteomic profile closer to the control reference and hence better biological recovery.

This score is not ad hoc: it directly reflects how closely a given mouse resembles the control proteomic state in the most biologically differentiating coordinates. The resulting groupwise summaries are given in Table~\ref{tab:recovery_summary}. Notably, among trisomic mice, the memantine-treated non-stimulated subgroup exhibits mean recovery score $-3.7440$, substantially better than the saline non-stimulated subgroup at $-4.9376$. In the stimulated trisomic subgroup, memantine also performs slightly better than saline, though the difference is much smaller.

\begin{table}[h!]
\centering
\caption{Recovery-score summary by genotype, treatment, and behavioral subgroup. Higher recovery score is better.}
\label{tab:recovery_summary}
\begin{tabular}{lllrccc}
\toprule
Genotype & Treatment & Behavior & $n$ & Mean recovery & SD recovery & Median recovery \\
\midrule
Control & Memantine & NotStimulated & 150 & -3.4351 & 0.9388 & -3.4112 \\
Control & Memantine & Stimulated & 150 & -3.8060 & 1.3515 & -3.5073 \\
Control & Saline & NotStimulated & 135 & -3.7100 & 1.9830 & -3.1175 \\
Control & Saline & Stimulated & 135 & -4.8579 & 2.8015 & -3.7499 \\
Trisomy & Memantine & NotStimulated & 135 & -3.7440 & 1.3092 & -3.4119 \\
Trisomy & Memantine & Stimulated & 135 & -4.6930 & 1.2108 & -4.5926 \\
Trisomy & Saline & NotStimulated & 135 & -4.9376 & 1.3441 & -4.9475 \\
Trisomy & Saline & Stimulated & 105 & -4.8154 & 1.5141 & -4.5037 \\
\bottomrule
\end{tabular}
\end{table}

These contrasts are visualized in Figures~\ref{fig:recovery_genotype_treatment}--\ref{fig:recovery_trisomy_violin}. The first two figures show the overall distribution of recovery scores by genotype/treatment and behavior/treatment, while the violin plot in Figure~\ref{fig:recovery_trisomy_violin} concentrates specifically on trisomic mice, where the scientific treatment question is most relevant.

\begin{figure}[h!]
\centering
\includegraphics[width=0.72\textwidth]{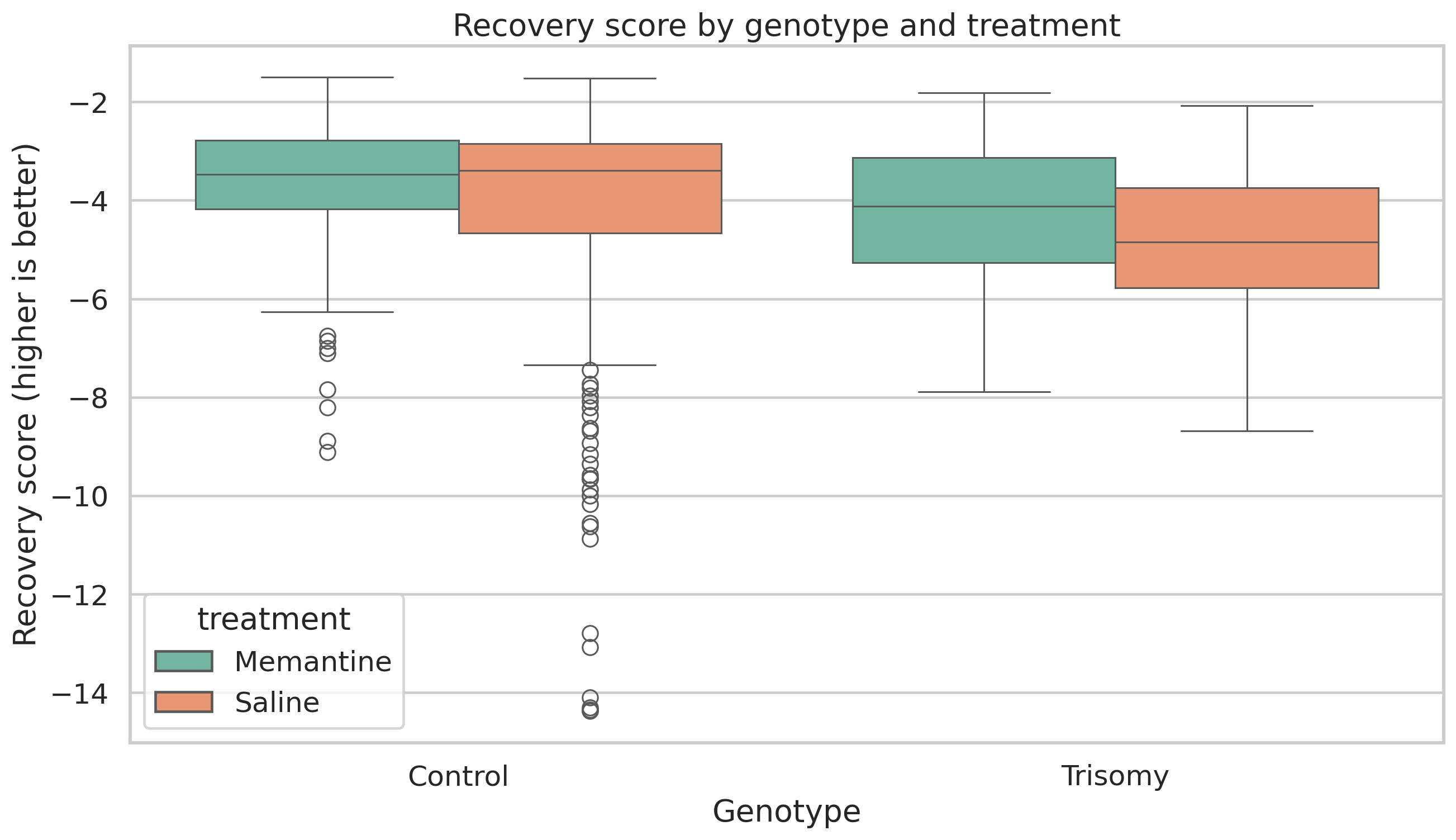}
\caption{Recovery score by genotype and treatment.}
\label{fig:recovery_genotype_treatment}
\end{figure}

\begin{figure}[h!]
\centering
\includegraphics[width=0.72\textwidth]{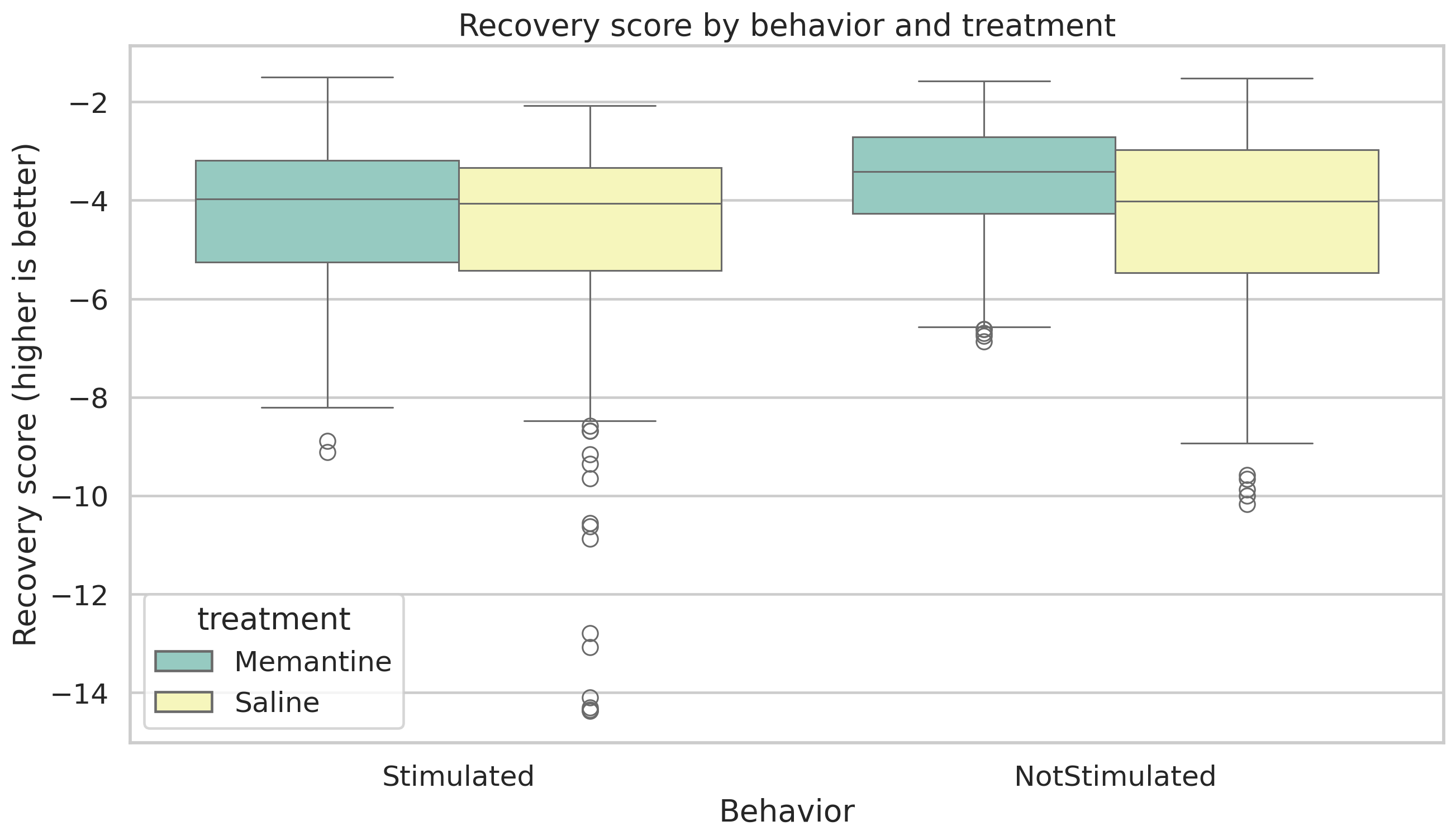}
\caption{Recovery score by behavioral condition and treatment.}
\label{fig:recovery_behavior_treatment}
\end{figure}

\begin{figure}[h!]
\centering
\includegraphics[width=0.72\textwidth]{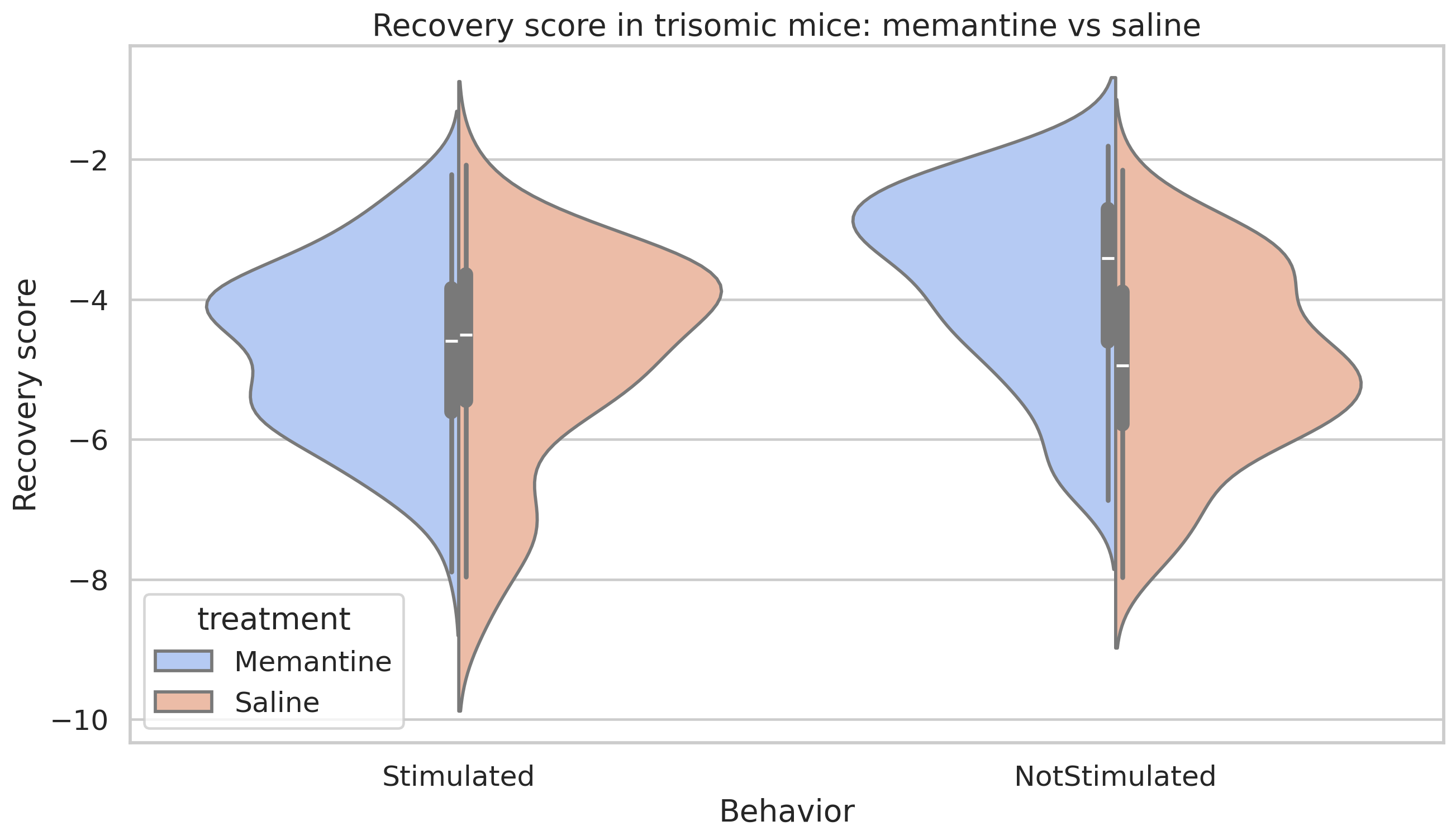}
\caption{Recovery score among trisomic mice by behavioral condition and treatment.}
\label{fig:recovery_trisomy_violin}
\end{figure}

\subsection{Bayesian credible-risk treatment-effect analysis}
\label{subsec:bayesian_realdata}

We now apply our proposed methodology. For each trisomic subgroup, we treated the difference
\[
\Delta = \mathbb{E}(\mathrm{RecoveryScore}\mid \text{Memantine})
-
\mathbb{E}(\mathrm{RecoveryScore}\mid \text{Saline})
\]
as the treatment effect of interest. Positive values of $\Delta$ imply that memantine shifts the proteomic profile closer to the control reference state.

Using an approximate conjugate Bayesian analysis for the subgroup means, we computed posterior means, posterior standard deviations, and $95\%$ credible intervals for $\Delta$. We then formed our proposed \emph{credible-risk score}
\[
\mathrm{CRS} = \widehat{\Delta}_{\mathrm{post}} - \kappa \,\widehat{\sigma}_{\mathrm{post}},
\qquad \kappa = 0.60,
\]
which penalizes uncertain positive effects and therefore yields a more conservative decision rule than posterior mean alone.

The results are given in Table~\ref{tab:subgroup_posterior}. For the \emph{Stimulated} trisomic subgroup, the posterior mean effect is only $0.1219$ with a $95\%$ credible interval $(-0.2324,\;0.4763)$, indicating weak evidence. However, the credible-risk score remains slightly positive at $0.0135$, so the method still yields a cautious memantine-favored decision. In contrast, for the \emph{NotStimulated} trisomic subgroup, the posterior mean effect is $1.1933$ with credible interval $(0.8769,\;1.5098)$, and the credible-risk score is $1.0965$, giving strong evidence in favor of memantine.

Thus, the real-data analysis suggests that the treatment benefit is not homogeneous: the proteomic improvement associated with memantine is especially pronounced in the trisomic non-stimulated subgroup.

\begin{table}[h!]
\centering
\caption{Posterior treatment effects of memantine versus saline among trisomic mice. The proposed method uses the credible-risk score $\mathrm{CRS}=\widehat{\Delta}_{\mathrm{post}}-\kappa\widehat{\sigma}_{\mathrm{post}}$ with $\kappa=0.60$.}
\label{tab:subgroup_posterior}
\begin{tabular}{lrrrrrrl}
\toprule
Subgroup & $n_{\mathrm{mem}}$ & $n_{\mathrm{sal}}$ & Post.\ mean & Post.\ SD & CI low & CI high & Recommendation \\
\midrule
Stimulated & 135 & 105 & 0.1219 & 0.1808 & -0.2324 & 0.4763 & Memantine favored \\
NotStimulated & 135 & 135 & 1.1933 & 0.1615 & 0.8769 & 1.5098 & Memantine favored \\
\bottomrule
\end{tabular}
\end{table}

Figure~\ref{fig:forest_subgroup} provides the corresponding forest plot. The visual contrast between the two subgroups is clear: the not-stimulated group shows a decisively positive and relatively precise posterior effect, while the stimulated group lies close to the boundary of practical uncertainty.

\begin{figure}[h!]
\centering
\includegraphics[width=0.72\textwidth]{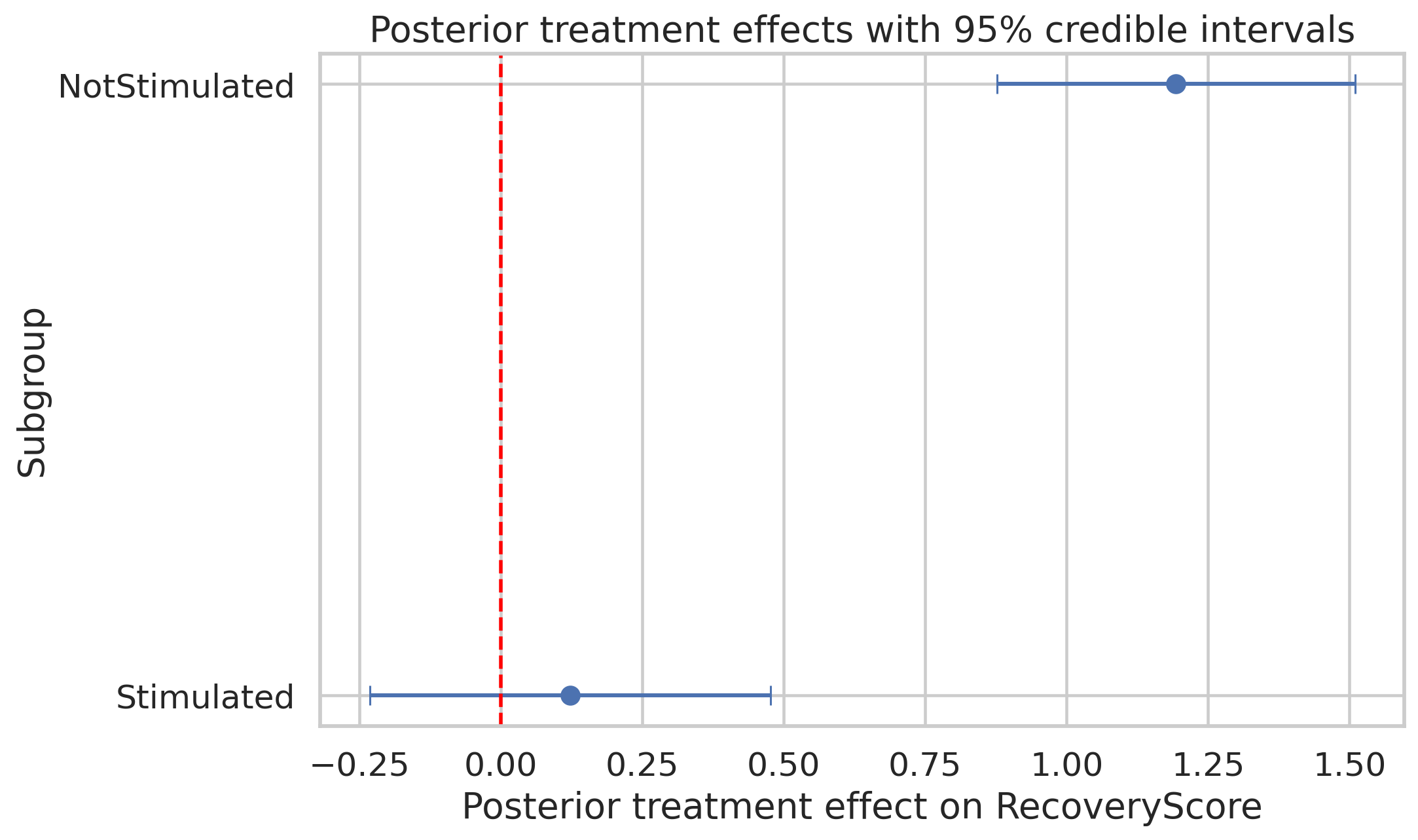}
\caption{Posterior treatment effects with $95\%$ credible intervals for the trisomic subgroups.}
\label{fig:forest_subgroup}
\end{figure}

\subsection{Sequential Bayesian learning}
\label{subsec:sequential_realdata}

A distinctive feature of our methodology is that it is not merely descriptive; it is \emph{sequentially adaptive}. To illustrate this, we processed the \emph{Trisomy + Stimulated} subgroup in random sequential order and updated the posterior treatment effect after each new observation. This produces a real-data analogue of the learning component in our general Bayesian strategic framework.

Figures~\ref{fig:sequential_effect} and \ref{fig:sequential_crs} show the resulting sequential learning behavior. The posterior effect and its uncertainty band evolve over time, while the credible-risk score tracks the stability of the treatment recommendation. These plots demonstrate that our method can operate naturally as an evidence-updating procedure, which is especially important in adaptive biological or biomedical experimentation.

\begin{figure}[h!]
\centering
\includegraphics[width=0.72\textwidth]{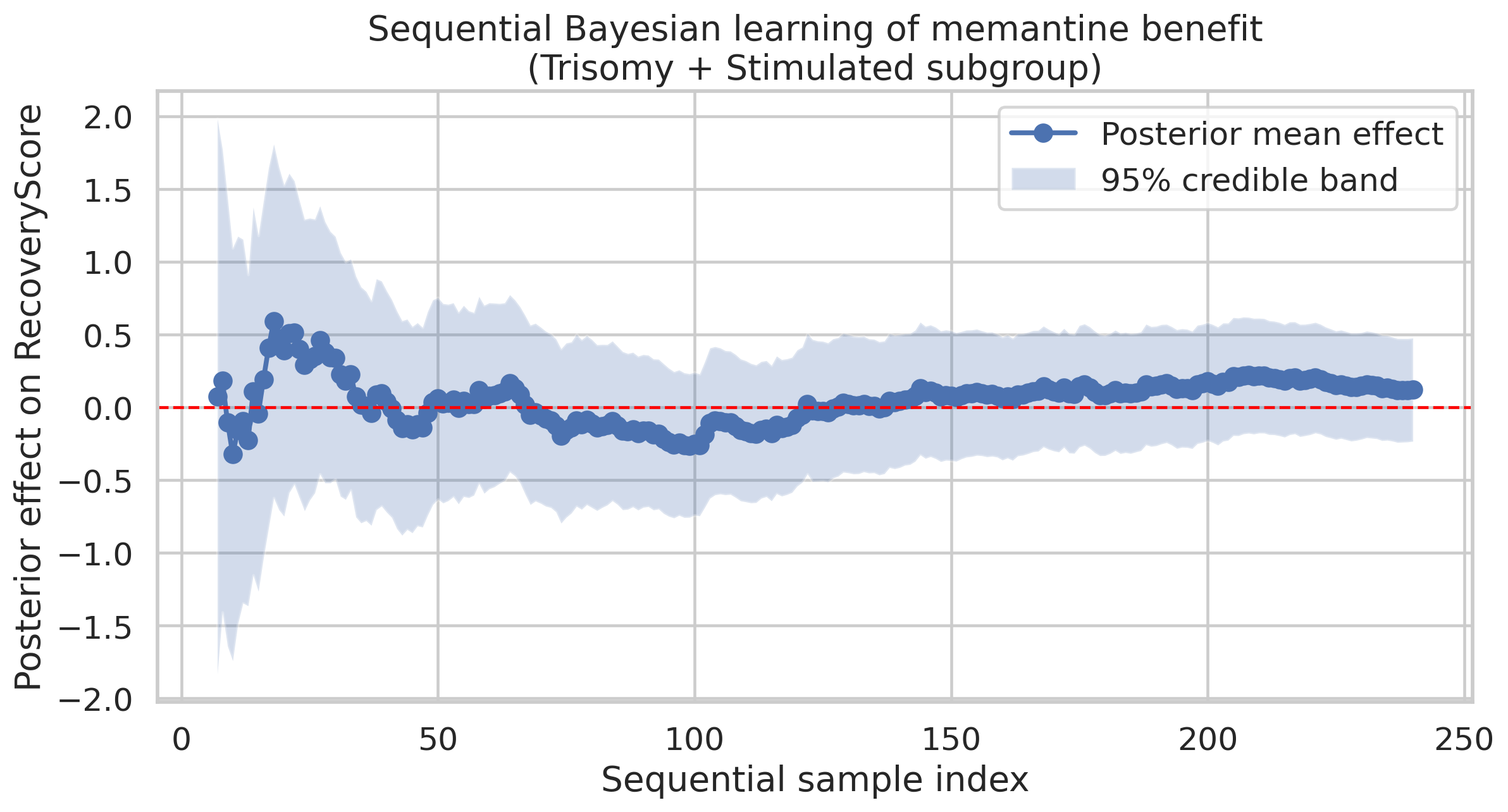}
\caption{Sequential Bayesian learning curve for the treatment effect in the trisomic stimulated subgroup.}
\label{fig:sequential_effect}
\end{figure}

\begin{figure}[h!]
\centering
\includegraphics[width=0.72\textwidth]{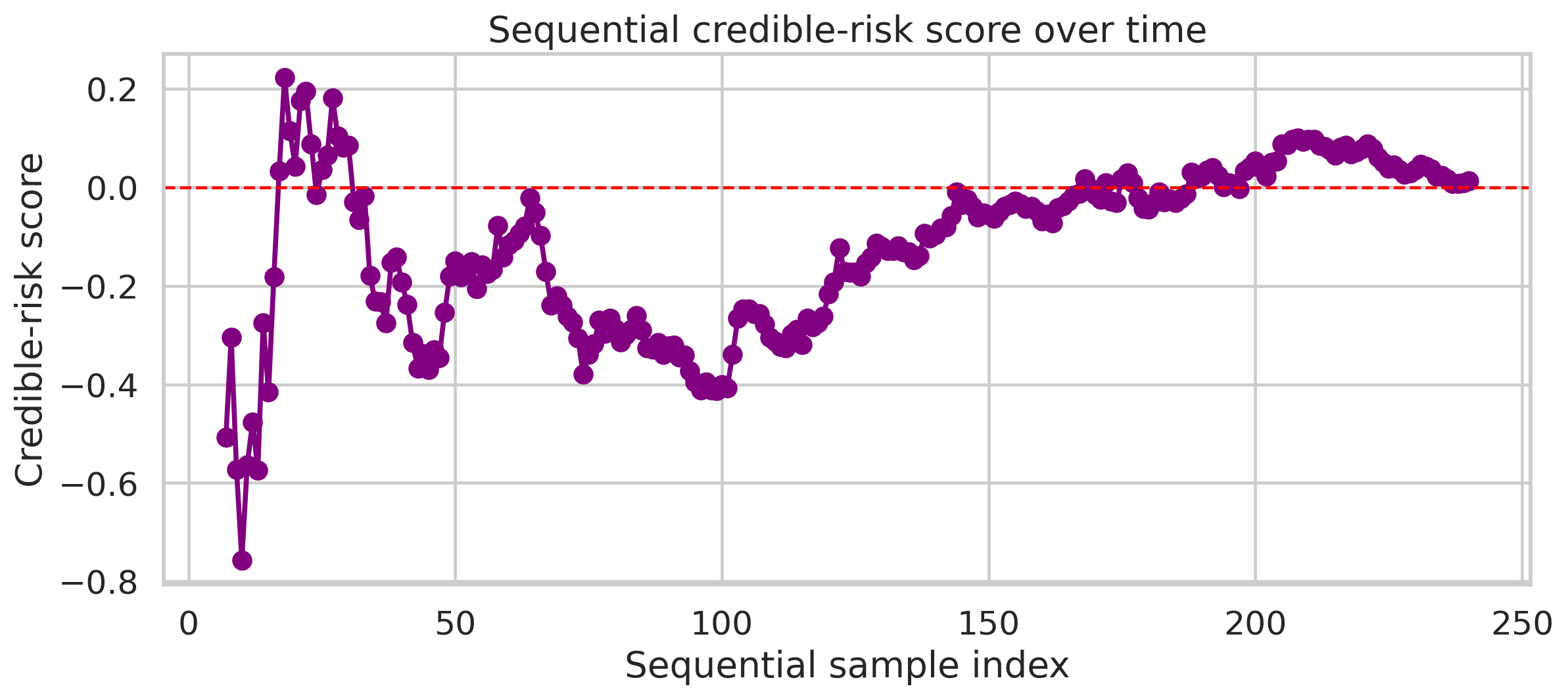}
\caption{Sequential credible-risk score for the treatment effect in the trisomic stimulated subgroup.}
\label{fig:sequential_crs}
\end{figure}

\subsection{Latent proteomic states}
\label{subsec:latent_states_realdata}

To enrich the empirical analysis beyond observed labels alone, we estimated latent proteomic states using PCA scores followed by $k$-means clustering. This reveals whether treatment effects differ across hidden molecular regimes rather than only across observed experimental categories.

The latent structure is shown in Figure~\ref{fig:latent_pca}. In the fitted solution, the trisomic mice appearing in the output were represented in two latent states with sufficient subgroup sample sizes for treatment-effect analysis. Table~\ref{tab:latent_state_effects} shows that both latent states display positive posterior memantine effects and positive credible-risk scores:
\[
0.5142 \quad \text{and} \quad 0.6947
\]
for the posterior mean effect, respectively. Hence, the beneficial memantine signal is not confined to a single latent configuration.

\begin{figure}[h!]
\centering
\includegraphics[width=0.72\textwidth]{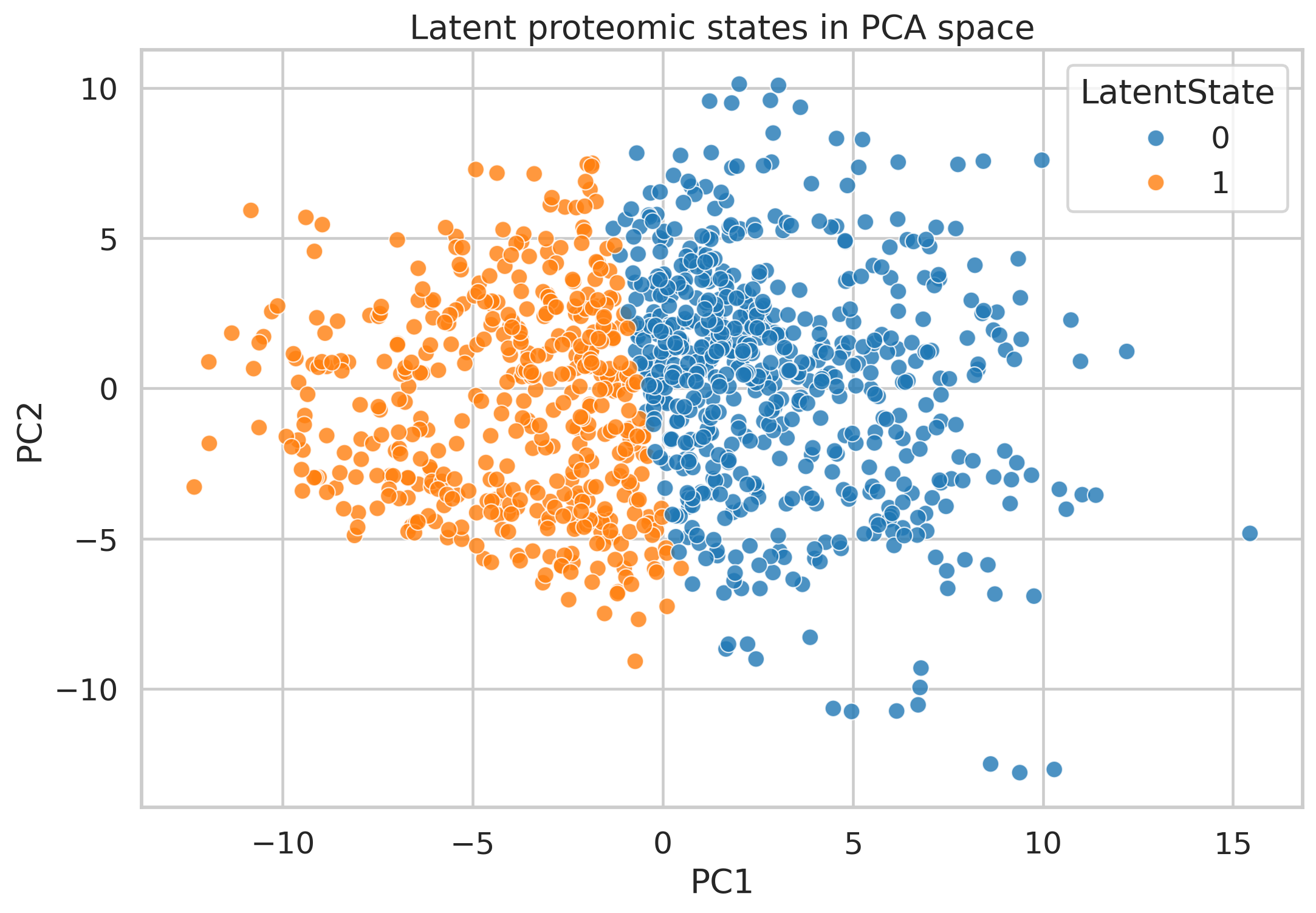}
\caption{Latent proteomic states represented in PCA space.}
\label{fig:latent_pca}
\end{figure}

\begin{table}[h!]
\centering
\caption{Posterior treatment effects within latent proteomic states among trisomic mice.}
\label{tab:latent_state_effects}
\begin{tabular}{rrrrrrrl}
\toprule
Latent state & $n_{\mathrm{mem}}$ & $n_{\mathrm{sal}}$ & Post.\ mean & Post.\ SD & CI low & CI high & Recommendation \\
\midrule
0 & 160 & 117 & 0.5142 & 0.1627 & 0.1953 & 0.8332 & Memantine favored \\
1 & 110 & 123 & 0.6947 & 0.1728 & 0.3560 & 1.0333 & Memantine favored \\
\bottomrule
\end{tabular}
\end{table}

The corresponding effect heatmap is shown in Figure~\ref{fig:latent_heatmap}, which provides a compact summary of both posterior mean effect and credible-risk score across latent states.

\begin{figure}[h!]
\centering
\includegraphics[width=0.65\textwidth]{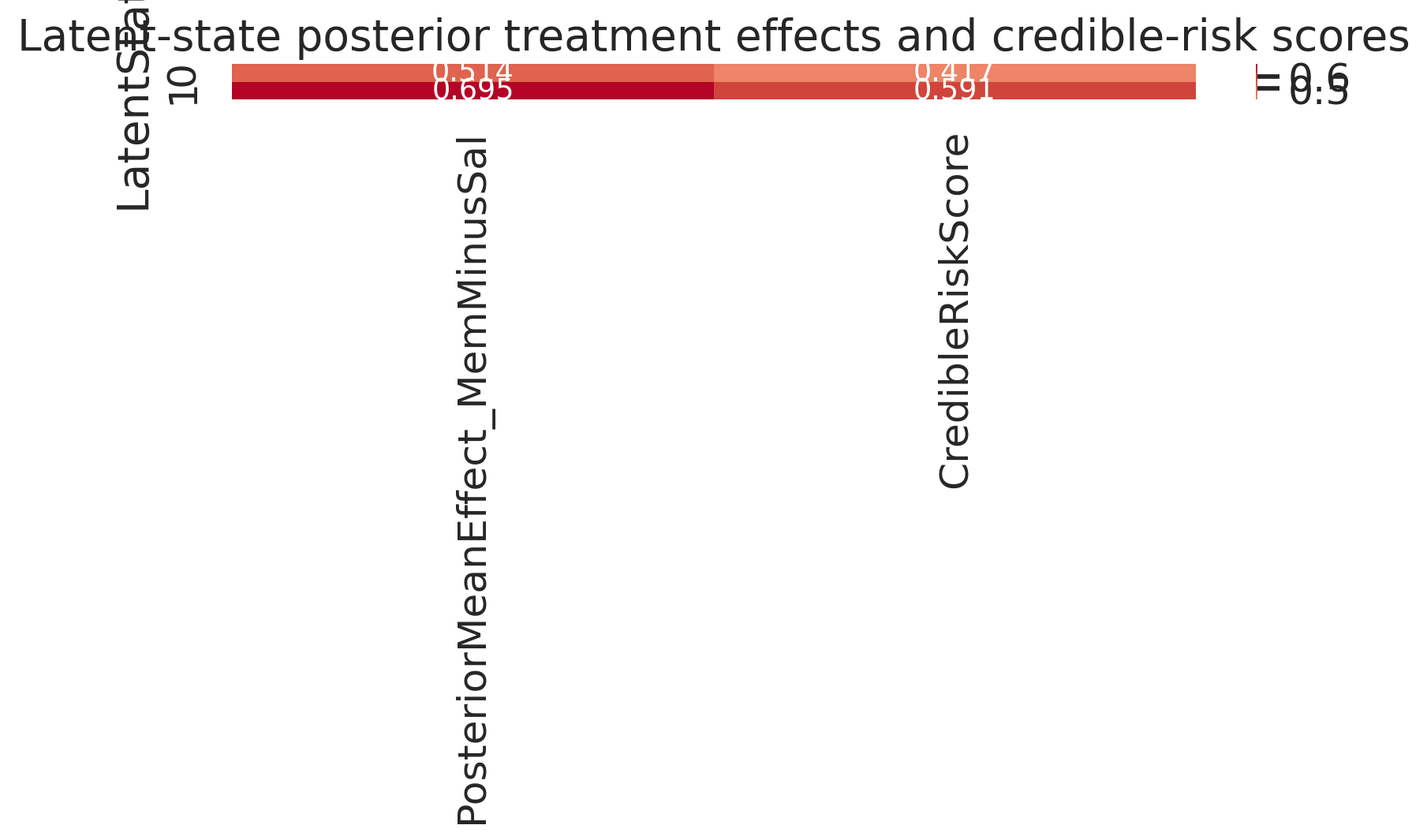}
\caption{Posterior treatment effects and credible-risk scores across latent proteomic states.}
\label{fig:latent_heatmap}
\end{figure}

\subsection{Protein-level recovery ranking}
\label{subsec:protein_level_realdata}

To move from global effect summaries to interpretable molecular markers, we ranked proteins according to how strongly memantine moved trisomic expression profiles toward the control-group mean. For a protein $j$, we computed the improvement score
\[
I_j
=
\left| \bar X^{(\mathrm{sal})}_j - \bar X^{(\mathrm{ctrl})}_j \right|
-
\left| \bar X^{(\mathrm{mem})}_j - \bar X^{(\mathrm{ctrl})}_j \right|,
\]
so that positive values indicate that memantine brings the trisomic mean closer to the control mean.

The leading proteins are listed in Table~\ref{tab:top_proteins}. The most prominent signal is \texttt{pPKCG\_N}, with improvement score $0.5106$ and bootstrap interval $(0.4114,\;0.6026)$, followed by \texttt{NR2A\_N} with improvement score $0.2698$ and interval $(0.1019,\;0.4217)$. Several additional proteins, including \texttt{S6\_N}, \texttt{pNR2B\_N}, \texttt{pMTOR\_N}, \texttt{NR1\_N}, and \texttt{pP70S6\_N}, also show positive evidence of movement toward the control proteomic profile under memantine. These patterns strengthen the biological interpretability of the proposed recovery-score analysis.

\begin{table}[h!]
\centering \small
\caption{Top proteins whose trisomic expression under memantine moves closer to the control-group mean. Positive values indicate better recovery toward the control reference.}
\label{tab:top_proteins}
\begin{tabular}{lrrrrrr}
\toprule
Protein & Control mean & Trisomy mem. mean & Trisomy sal. mean & Improvement & CI low & CI high \\
\midrule
pPKCG\_N & 1.5570 & 1.6336 & 2.1441 & 0.5106 & 0.4114 & 0.6026 \\
NR2A\_N & 3.9847 & 3.8111 & 3.5413 & 0.2698 & 0.1019 & 0.4217 \\
S6\_N & 0.3847 & 0.4232 & 0.5419 & 0.1187 & 0.0949 & 0.1386 \\
pNR2B\_N & 1.5933 & 1.5751 & 1.4711 & 0.1040 & 0.0553 & 0.1433 \\
pMTOR\_N & 0.7792 & 0.7731 & 0.6950 & 0.0782 & 0.0540 & 0.0978 \\
pELK\_N & 1.4432 & 1.4727 & 1.3422 & 0.0714 & -0.0278 & 0.1491 \\
NR1\_N & 2.3337 & 2.2880 & 2.2207 & 0.0673 & 0.0107 & 0.1225 \\
pPKCAB\_N & 1.5159 & 1.4471 & 1.6352 & 0.0504 & -0.0544 & 0.1620 \\
pP70S6\_N & 0.3687 & 0.4006 & 0.4508 & 0.0502 & 0.0196 & 0.0825 \\
P38\_N & 0.4299 & 0.4209 & 0.3743 & 0.0466 & 0.0313 & 0.0595 \\
NR2B\_N & 0.5806 & 0.5672 & 0.5269 & 0.0403 & 0.0231 & 0.0567 \\
pNR2A\_N & 0.7674 & 0.7001 & 0.6602 & 0.0398 & 0.0104 & 0.0720 \\
AcetylH3K9\_N & 0.1738 & 0.2456 & 0.2851 & 0.0395 & 0.0058 & 0.0720 \\
IL1B\_N & 0.5362 & 0.5357 & 0.4970 & 0.0387 & 0.0189 & 0.0461 \\
Tau\_N & 0.1927 & 0.2128 & 0.2500 & 0.0372 & 0.0259 & 0.0494 \\
CaNA\_N & 1.3240 & 1.3372 & 1.3710 & 0.0338 & -0.0285 & 0.0718 \\
pNR1\_N & 0.8465 & 0.8168 & 0.7867 & 0.0301 & 0.0103 & 0.0494 \\
pCASP9\_N & 1.5520 & 1.5667 & 1.5189 & 0.0184 & -0.0340 & 0.0619 \\
GluR3\_N & 0.2304 & 0.2204 & 0.2032 & 0.0172 & 0.0121 & 0.0221 \\
MTOR\_N & 0.4664 & 0.4438 & 0.4291 & 0.0147 & 0.0020 & 0.0275 \\
\bottomrule
\end{tabular}
\end{table}

The forest representation in Figure~\ref{fig:protein_forest} summarizes these leading recovery-associated proteins, while Figure~\ref{fig:protein_corr} shows their correlation structure. These two figures together indicate that the empirical effect of memantine is not driven by a single isolated variable, but by a coordinated set of molecular shifts.

\begin{figure}[h!]
\centering
\includegraphics[width=0.78\textwidth]{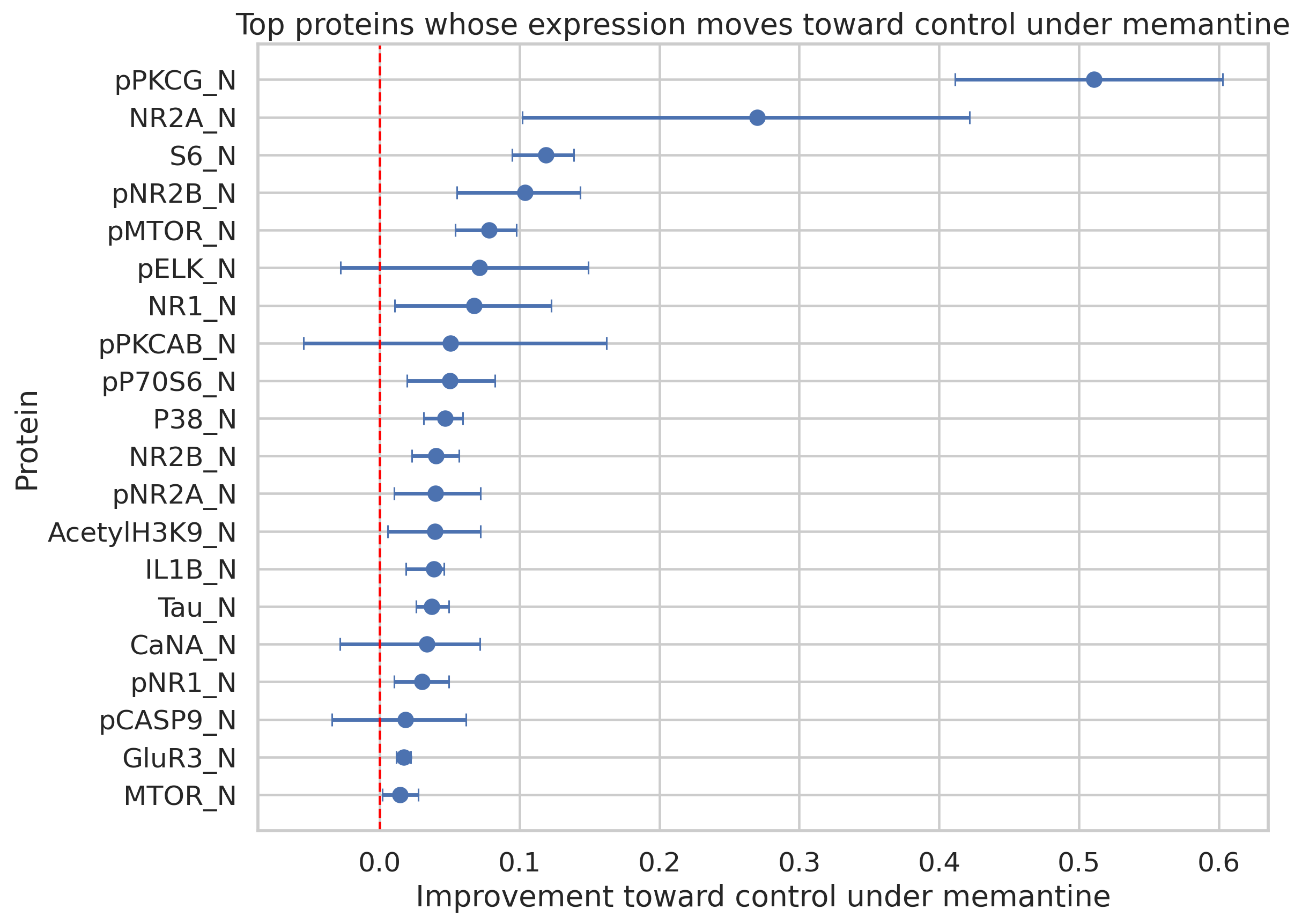}
\caption{Top proteins whose expression under memantine moves trisomic mice toward the control reference profile.}
\label{fig:protein_forest}
\end{figure}

\begin{figure}[h!]
\centering
\includegraphics[width=0.72\textwidth]{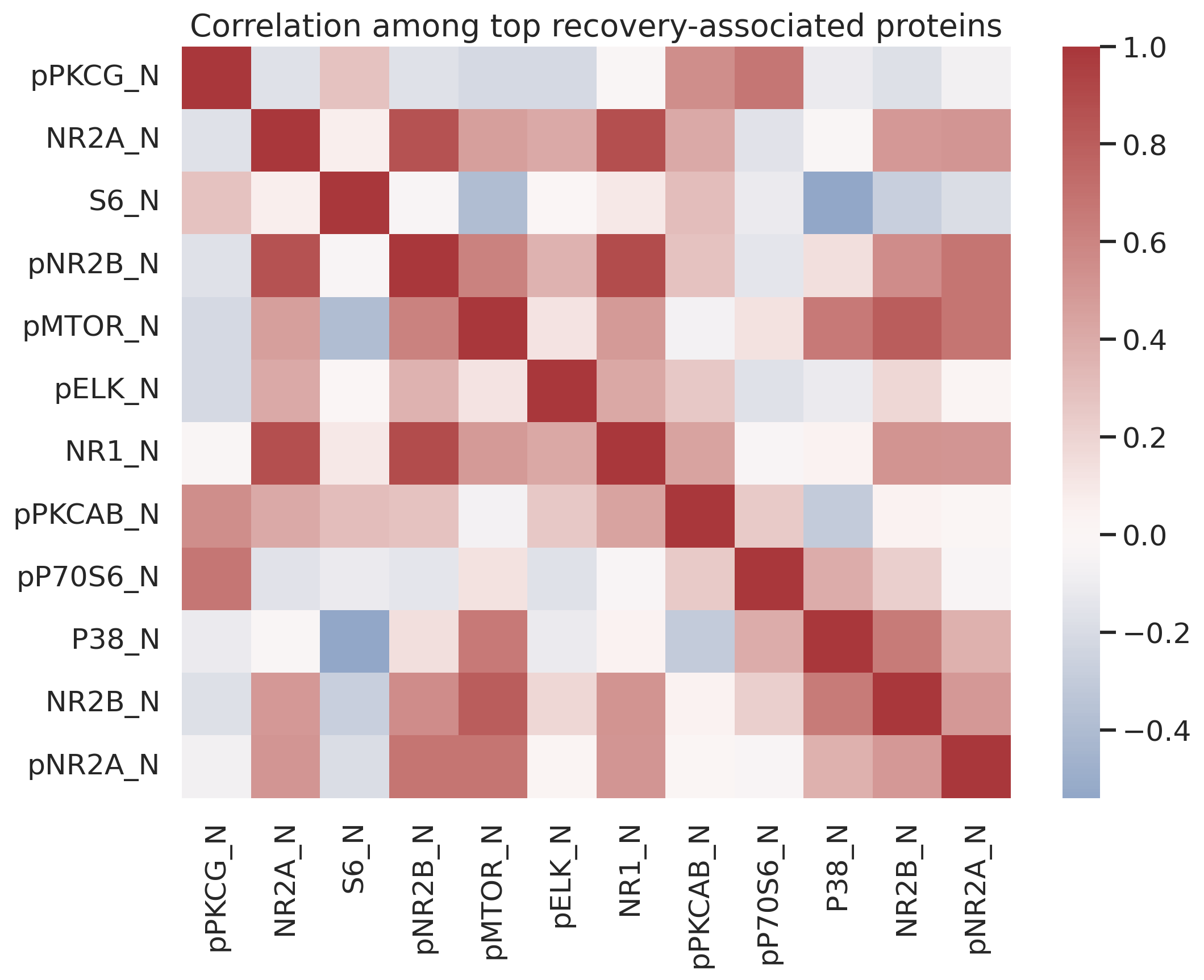}
\caption{Correlation structure among the top recovery-associated proteins.}
\label{fig:protein_corr}
\end{figure}

\subsection{Global treatment-effect summary}
\label{subsec:global_realdata}

Finally, Table~\ref{tab:global_effects} summarizes the global treatment effect separately in the control and trisomy groups. The posterior mean effect of memantine on the recovery score is positive in both groups:
\[
0.6626 \quad \text{for Control}, \qquad
0.6655 \quad \text{for Trisomy}.
\]
The bootstrap confidence intervals are also fully positive in both cases. Importantly, the credible-risk scores remain positive for both groups, indicating that the estimated benefit is not merely a posterior-mean artifact but survives conservative uncertainty adjustment.

\begin{table}[h!]
\centering
\caption{Global treatment-effect summary by genotype. Positive effects favor memantine.}
\label{tab:global_effects}
\begin{tabular}{lrrrrrrrrrr}
\toprule
Group & $n_{\mathrm{mem}}$ & $n_{\mathrm{sal}}$ & Post.\ mean & Post.\ SD & Post.\ CI low & Post.\ CI high & Boot.\ mean & Boot.\ CI low & Boot.\ CI high & CRS \\
\midrule
Control & 300 & 270 & 0.6626 & 0.1660 & 0.3372 & 0.9881 & 0.6717 & 0.3501 & 1.0048 & 0.5630 \\
Trisomy & 270 & 240 & 0.6655 & 0.1229 & 0.4247 & 0.9063 & 0.6649 & 0.4267 & 0.9228 & 0.5918 \\
\bottomrule
\end{tabular}
\end{table}

The corresponding barplot is shown in Figure~\ref{fig:global_barplot}. The near equality of the global posterior effects across control and trisomy groups is interesting in itself, but the subgroup analyses above reveal that the internal composition of this effect is heterogeneous and scientifically richer than a single overall average would suggest.

\begin{figure}[h!]
\centering
\includegraphics[width=0.70\textwidth]{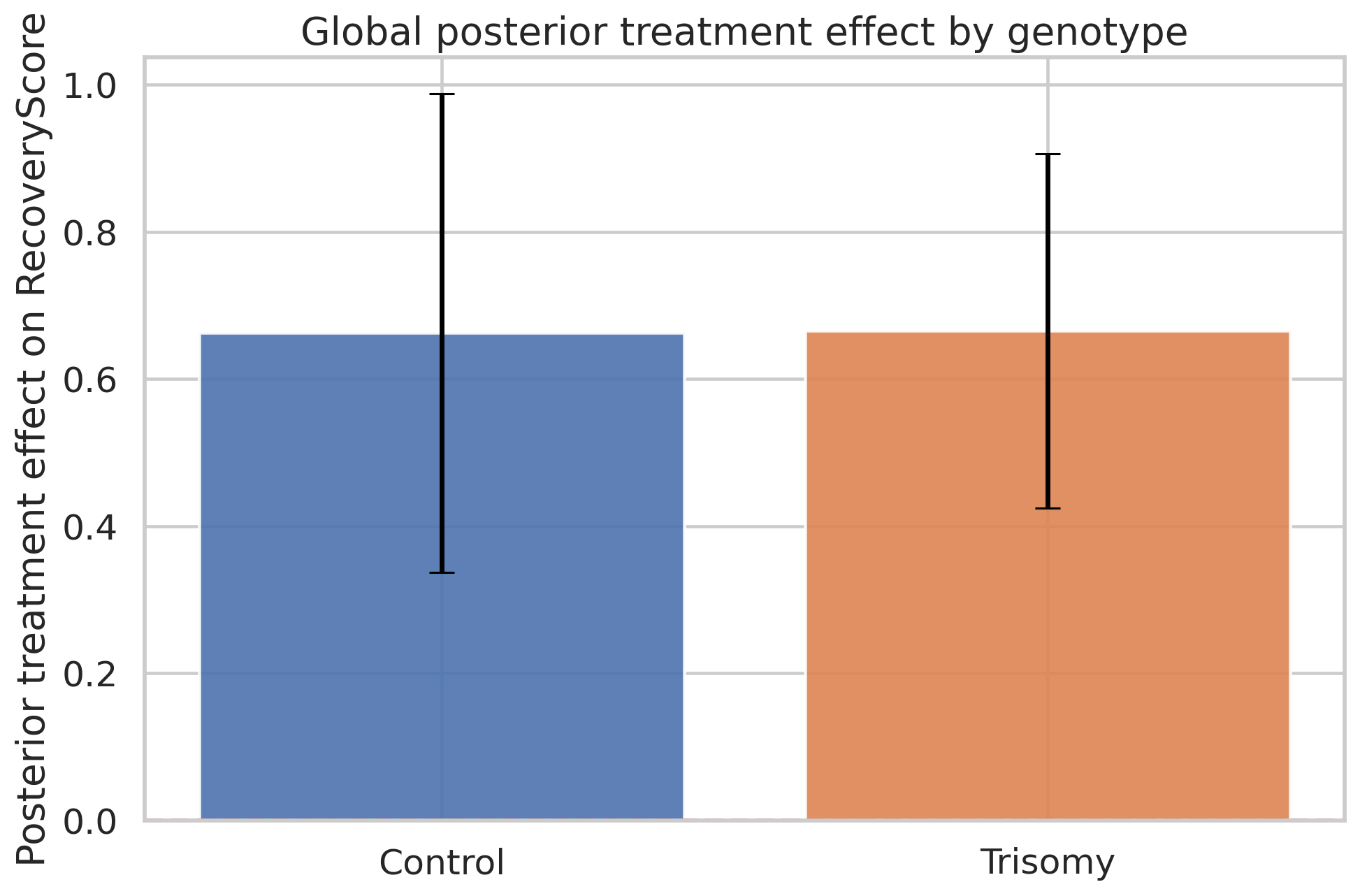}
\caption{Global posterior treatment effect of memantine on the recovery score by genotype.}
\label{fig:global_barplot}
\end{figure}

\begin{remark}[Key findings from the real-data analysis]
\label{rem:realdata_keyfindings}
The real-data analysis demonstrates that the proposed Bayesian credible-risk methodology is not merely a theoretical construction; it yields a coherent and biologically interpretable empirical workflow on a high-dimensional proteomic dataset. The method contributes value in several ways. First, it constructs a principled recovery score by anchoring inference to the control reference proteomic state. Second, it quantifies subgroup-specific memantine effects with posterior uncertainty, rather than relying only on point summaries. Third, by using the credible-risk score, it converts posterior uncertainty into conservative scientific decision-making. The strongest empirical signal appears in the trisomic non-stimulated subgroup, where the posterior memantine benefit is both large and precise. In addition, the latent-state analysis shows that the favorable memantine signal persists across hidden proteomic regimes, while the protein-level ranking identifies specific markers such as \texttt{pPKCG\_N}, \texttt{NR2A\_N}, \texttt{S6\_N}, and \texttt{pNR2B\_N} as particularly responsive toward the control reference. Altogether, this real-data study supports the practical worth of our proposed methodology as an uncertainty-aware inferential tool for complex biological data.
\end{remark}

\section{Overall Discussion}
\label{sec:overall_discussion}

Taken together, the theoretical development, simulation evidence, and real-data illustration clarify the role of the proposed methodology.

From a modeling perspective, the paper shows that competitive operational decisions can be embedded naturally in a Bayesian dynamic game when uncertainty about the market and uncertainty about rivals coexist. The central conceptual contribution is that beliefs are not passive summaries of information; they are active state variables that shape equilibrium behavior. This makes the framework especially suitable for repeated decision environments in which learning and competition unfold simultaneously.

From a computational and empirical perspective, the simulation study indicates that Bayesian learning is indispensable in the proposed competitive setting. The main performance gain comes from adaptively updating market and rival information rather than relying on static prior-driven decisions. Within the class of Bayesian learners, the credible-risk rule is best interpreted as a conservative regularizer: it slightly improves profit-based performance in the simulation while preserving broadly similar learning quality. Thus, its value lies less in dramatic numerical dominance and more in producing disciplined decisions under uncertainty.

The real-data analysis complements the simulation study by showing that the credible-risk principle is useful beyond the stylized inventory game itself. On the proteomic dataset, it yields interpretable subgroup, latent-state, and marker-level findings. This strengthens the methodological message of the paper: the proposed framework is valuable not only as a formal equilibrium model, but also as a general uncertainty-aware Bayesian decision paradigm.

At the same time, the paper should be interpreted with appropriate caution. The computational scheme is approximate, the simulation design is stylized, and the real-data illustration demonstrates the inferential principle rather than a literal market game. These limitations do not weaken the contribution; instead, they define a transparent boundary between what has been established here and what remains open for further development.

\section{Conclusion and Future Work}
\label{sec:conclusion}

This paper develops a hierarchical Bayesian dynamic game for competitive inventory and pricing under incomplete information. The proposed framework integrates private rival characteristics, Bayesian learning about demand and substitution, and a credible-risk decision rule that transforms posterior uncertainty into conservative strategic behavior. In this way, the paper builds a bridge between Bayesian game theory and operations research and provides a unified language for learning-driven competition.

The simulation study shows that posterior learning is central to strong performance in repeated competitive environments, while the credible-risk adjustment serves as a useful operational regularizer. The real-data illustration further demonstrates that the same uncertainty-aware Bayesian principle can support interpretable analysis in complex high-dimensional biological data. Together, these results suggest that the proposed methodology is both mathematically meaningful and practically useful.

Several directions remain for future work. One important extension is to strengthen the computational and theoretical treatment of equilibrium under richer state spaces and more general demand models. Another is to move from duopoly to larger multi-agent settings, including networked or mean-field competition. It would also be valuable to study alternative uncertainty penalties, richer censored-demand learning schemes, and formal welfare or information-disclosure questions. Finally, domain-specific applications in supply chains, platform markets, and adaptive experimental systems offer promising directions for translating the present framework into substantive practice.

\section*{Declarations}

\noindent\textbf{Ethics approval and consent to participate.}
Not applicable.

\medskip
\noindent\textbf{Consent for publication.}
Not applicable.

\medskip
\noindent\textbf{Availability of data and materials.}
The real dataset analyzed in this paper is publicly available through OpenML and is also documented in the UCI Machine Learning Repository. The simulation data were generated by the author using the procedures described in the paper, and is available inside the code repository.

\medskip
\noindent\textbf{Code availability.}
All code for reproducing the experiments and generating the figures/tables is publicly available at:

\url{https://github.com/debashisdotchatterjee/Bayesian-Dynamic-Game-1}.

\noindent\textbf{Funding.}
The author received no specific funding for this work.

\noindent\textbf{Competing interests.}
The author declares that there are no competing interests.

\appendix
\section{Appendix: Formal Proofs Under Strengthened Regularity Conditions}
\label{app:proofs}

In the main text, the existence results were stated at a high level in order to emphasize the modeling structure. To make the proofs mathematically rigorous, we now record strengthened regularity assumptions that are sufficient for the corresponding fixed-point and dynamic programming arguments. The proofs given below are fully valid under these strengthened assumptions.

\subsection{Strengthened assumptions}

Let $\mathcal{X}_i$ denote player $i$'s belief-state space and let $\mathcal{A}_i$ denote the action space.

\begin{assumption}[Compact metric action spaces]
\label{ass:compact_metric}
For each player $i$, the action space $\mathcal{A}_i$ is a nonempty compact metric space.
\end{assumption}

\begin{assumption}[Compact metric state spaces]
\label{ass:compact_state}
For each player $i$, the state space $\mathcal{X}_i$ is a nonempty compact metric space.
\end{assumption}

\begin{assumption}[Bounded continuous one-period return]
\label{ass:bounded_return}
For each player $i$, the one-period credible-risk return
\[
r_i(x_i,a_i,a_j)
\]
is real-valued, bounded, and jointly continuous on
\[
\mathcal{X}_i \times \mathcal{A}_i \times \mathcal{A}_j.
\]
\end{assumption}

\begin{assumption}[Weakly continuous transition kernel]
\label{ass:weak_transition}
For each player $i$, the controlled transition kernel
\[
Q_i(\cdot \mid x_i,a_i,a_j)
\]
from $\mathcal{X}_i \times \mathcal{A}_i \times \mathcal{A}_j$ into probability measures on $\mathcal{X}_i$ is weakly continuous.
\end{assumption}

\begin{assumption}[Discount factor]
\label{ass:discount}
The discount factor satisfies $0<\delta<1$.
\end{assumption}

\begin{assumption}[Mixed strategies]
\label{ass:mixed}
Players are allowed to use mixed actions. The set $\Delta(\mathcal{A}_i)$ of Borel probability measures on $\mathcal{A}_i$ is endowed with the weak topology.
\end{assumption}

\begin{assumption}[Quasi-concavity / convexification]
\label{ass:convex}
Either the action spaces are convex and the relevant objective function is quasi-concave in own action, or mixed strategies are used so that best-response sets are convex-valued.
\end{assumption}

\begin{remark}
Assumptions \ref{ass:compact_metric}--\ref{ass:convex} are stronger than the streamlined assumptions used in the main body, but they are standard and sufficient for rigorous existence arguments in discounted stochastic games with compact state and action spaces.
\end{remark}

\subsection{A rigorous best-response existence result}

We first restate the interim best-response result in a form that admits a complete proof.

\begin{proposition}[Existence of interim best responses]
\label{prop:bestresponse_appendix}
Suppose Assumptions \ref{ass:compact_metric}, \ref{ass:compact_state}, \ref{ass:bounded_return}, \ref{ass:weak_transition}, and \ref{ass:discount} hold. Fix player $j$'s measurable strategy $\sigma_j$. Let $V_j$ be a bounded continuous continuation rule for player $i$. Then, for every state $x_i \in \mathcal{X}_i$, player $i$'s interim optimization problem
\[
\sup_{a_i \in \mathcal{A}_i}
\left\{
r_i\bigl(x_i,a_i,\sigma_j(x_j)\bigr)
+
\delta \int_{\mathcal{X}_i} V_i(x_i')\,Q_i(dx_i' \mid x_i,a_i,\sigma_j(x_j))
\right\}
\]
admits at least one maximizer.
\end{proposition}

\begin{proof}
Fix $x_i \in \mathcal{X}_i$ and define
\[
\Phi_i(a_i)
=
r_i\bigl(x_i,a_i,\sigma_j(x_j)\bigr)
+
\delta \int_{\mathcal{X}_i} V_i(x_i')\,Q_i(dx_i' \mid x_i,a_i,\sigma_j(x_j)).
\]
We show that $\Phi_i$ is continuous on the compact set $\mathcal{A}_i$.

By Assumption \ref{ass:bounded_return}, the map
\[
a_i \mapsto r_i\bigl(x_i,a_i,\sigma_j(x_j)\bigr)
\]
is continuous. By Assumption \ref{ass:weak_transition}, if $a_i^{(n)} \to a_i$ in $\mathcal{A}_i$, then
\[
Q_i(\cdot \mid x_i,a_i^{(n)},\sigma_j(x_j))
\Rightarrow
Q_i(\cdot \mid x_i,a_i,\sigma_j(x_j))
\]
weakly. Since $V_i$ is bounded and continuous, the portmanteau theorem yields
\[
\int_{\mathcal{X}_i} V_i(x_i')\,Q_i(dx_i' \mid x_i,a_i^{(n)},\sigma_j(x_j))
\to
\int_{\mathcal{X}_i} V_i(x_i')\,Q_i(dx_i' \mid x_i,a_i,\sigma_j(x_j)).
\]
Hence $\Phi_i$ is continuous.

Because $\mathcal{A}_i$ is compact by Assumption \ref{ass:compact_metric}, the Weierstrass extreme value theorem implies that $\Phi_i$ attains its maximum on $\mathcal{A}_i$. Therefore at least one maximizing action exists.
\end{proof}

\subsection{Bellman operator and value existence}

We now show that the player's Bellman operator is a contraction on bounded continuous functions.

Let $C_b(\mathcal{X}_i)$ denote the Banach space of bounded continuous real-valued functions on $\mathcal{X}_i$, equipped with the sup norm
\[
\|V\|_\infty = \sup_{x_i \in \mathcal{X}_i} |V(x_i)|.
\]

For fixed rival strategy $\sigma_j$, define the Bellman operator $T_i:C_b(\mathcal{X}_i)\to C_b(\mathcal{X}_i)$ by
\[
(T_i V)(x_i)
=
\sup_{a_i\in\mathcal{A}_i}
\left\{
r_i\bigl(x_i,a_i,\sigma_j(x_j)\bigr)
+
\delta \int_{\mathcal{X}_i} V(x_i')\,Q_i(dx_i' \mid x_i,a_i,\sigma_j(x_j))
\right\}.
\]

\begin{lemma}[Bellman operator preserves bounded continuity]
\label{lem:preserve}
Under Assumptions \ref{ass:compact_metric}--\ref{ass:discount}, if $V \in C_b(\mathcal{X}_i)$ then $T_iV \in C_b(\mathcal{X}_i)$.
\end{lemma}

\begin{proof}
Fix $V \in C_b(\mathcal{X}_i)$. For each $(x_i,a_i)$ define
\[
G(x_i,a_i)
=
r_i\bigl(x_i,a_i,\sigma_j(x_j)\bigr)
+
\delta \int_{\mathcal{X}_i} V(x_i')\,Q_i(dx_i' \mid x_i,a_i,\sigma_j(x_j)).
\]
By Assumptions \ref{ass:bounded_return} and \ref{ass:weak_transition}, the function $G$ is jointly continuous in $(x_i,a_i)$. Since $\mathcal{A}_i$ is compact, the maximum theorem implies that
\[
x_i \mapsto \sup_{a_i \in \mathcal{A}_i} G(x_i,a_i)
\]
is continuous. Boundedness follows from boundedness of $r_i$ and $V$. Therefore $T_iV \in C_b(\mathcal{X}_i)$.
\end{proof}

\begin{lemma}[Bellman operator is a contraction]
\label{lem:contraction}
Under Assumptions \ref{ass:compact_metric}--\ref{ass:discount}, the operator $T_i$ is a contraction on $C_b(\mathcal{X}_i)$ with modulus $\delta$.
\end{lemma}

\begin{proof}
Let $V,W \in C_b(\mathcal{X}_i)$ and fix $x_i \in \mathcal{X}_i$. For every $a_i \in \mathcal{A}_i$,
\[
\left|
\int V(x_i')\,Q_i(dx_i' \mid x_i,a_i,\sigma_j(x_j))
-
\int W(x_i')\,Q_i(dx_i' \mid x_i,a_i,\sigma_j(x_j))
\right|
\le
\|V-W\|_\infty.
\]
Therefore,
\[
\left|
r_i(\cdot)+\delta \int V\,dQ_i
-
\Bigl(r_i(\cdot)+\delta \int W\,dQ_i\Bigr)
\right|
\le
\delta \|V-W\|_\infty.
\]
Taking suprema over $a_i$ preserves this bound, yielding
\[
|(T_iV)(x_i)-(T_iW)(x_i)|
\le
\delta \|V-W\|_\infty.
\]
Now take the supremum over $x_i \in \mathcal{X}_i$ to obtain
\[
\|T_iV-T_iW\|_\infty \le \delta \|V-W\|_\infty.
\]
Since $0<\delta<1$, $T_i$ is a contraction.
\end{proof}

\begin{corollary}[Existence and uniqueness of the value function]
\label{cor:value}
Under Assumptions \ref{ass:compact_metric}--\ref{ass:discount}, for fixed rival strategy $\sigma_j$, there exists a unique value function $V_i^\ast \in C_b(\mathcal{X}_i)$ satisfying
\[
T_i V_i^\ast = V_i^\ast.
\]
\end{corollary}

\begin{proof}
By Lemma \ref{lem:preserve}, $T_i$ maps $C_b(\mathcal{X}_i)$ into itself. By Lemma \ref{lem:contraction}, $T_i$ is a contraction on the complete metric space $C_b(\mathcal{X}_i)$. The Banach fixed-point theorem therefore yields a unique fixed point.
\end{proof}

\subsection{A rigorous equilibrium existence theorem}

We now state a strengthened equilibrium existence result. This theorem is the mathematically rigorous version of the equilibrium-existence claim used in the main text.

\begin{theorem}[Existence of a credible-risk Markov perfect Bayesian equilibrium under strengthened assumptions]
\label{thm:equilibrium_appendix}
Suppose Assumptions \ref{ass:compact_metric}--\ref{ass:convex} hold. Suppose further that:
\begin{enumerate}[label=(\roman*)]
    \item for each player, the continuation value is determined by the fixed point of the Bellman operator described above;
    \item best-response correspondences are nonempty for every rival strategy;
    \item best-response correspondences are upper hemicontinuous in the product weak topology on mixed strategy spaces;
    \item best-response correspondences are convex-valued.
\end{enumerate}
Then there exists at least one mixed-strategy credible-risk Markov perfect Bayesian equilibrium.
\end{theorem}

\begin{proof}
Let
\[
\Sigma_i
\]
denote player $i$'s set of measurable mixed Markov strategies from $\mathcal{X}_i$ into $\Delta(\mathcal{A}_i)$. Under Assumptions \ref{ass:compact_metric}, \ref{ass:compact_state}, and \ref{ass:mixed}, the mixed action spaces are compact in the weak topology by Prokhorov's theorem, and the corresponding strategy spaces are compact in the product weak topology under the usual measurable-selection framework.

Define the best-response correspondence
\[
\mathcal{B}_i:\Sigma_j \rightrightarrows \Sigma_i
\]
by assigning to each rival strategy the set of player-$i$ strategies that maximize the discounted credible-risk value. By Proposition \ref{prop:bestresponse_appendix} and Corollary \ref{cor:value}, each $\mathcal{B}_i(\sigma_j)$ is nonempty. By assumption, each $\mathcal{B}_i$ is upper hemicontinuous and convex-valued.

Now define the joint best-response correspondence
\[
\mathcal{B}(\sigma_1,\sigma_2)
=
\mathcal{B}_1(\sigma_2)\times \mathcal{B}_2(\sigma_1)
\]
from $\Sigma_1\times\Sigma_2$ into itself. The product strategy space is compact and convex, and $\mathcal{B}$ has nonempty, convex, compact values and is upper hemicontinuous. Therefore Kakutani's fixed-point theorem applies, yielding a point
\[
(\sigma_1^\ast,\sigma_2^\ast)\in \Sigma_1\times\Sigma_2
\]
such that
\[
(\sigma_1^\ast,\sigma_2^\ast)\in \mathcal{B}(\sigma_1^\ast,\sigma_2^\ast).
\]
By construction, each player's strategy is a best response to the other's and is measurable with respect to the current belief state. Hence the profile is a mixed-strategy credible-risk Markov perfect Bayesian equilibrium.
\end{proof}



\bibliographystyle{plainnat}
\bibliography{bayesian_game_inventory}

\end{document}